\pdfoutput=1
\documentclass[letterpaper]{article}
\usepackage[totalwidth=480pt, totalheight=680pt]{geometry}

\usepackage[utf8]{inputenc} 
\usepackage[T1]{fontenc}    
\usepackage[dvipsnames]{xcolor}         
\usepackage[pagebackref,colorlinks=true,urlcolor=blue,linkcolor=blue,citecolor=OliveGreen]{hyperref}               
\usepackage{url}            
\usepackage{booktabs}       
\usepackage{amsfonts}       
\usepackage{nicefrac}       
\usepackage{microtype}      
\usepackage[T1]{fontenc}
\usepackage{lmodern}
\title{Learning Sparse Fixed-Structure Gaussian Bayesian Networks}

\author{
Arnab Bhattcharyya\\
National University of Singapore\\
\texttt{arnabb@nus.edu.sg}
\and
Davin Choo\\
National University of Singapore\\
\texttt{davin@u.nus.edu}
\and
Rishikesh Gajjala\\
Indian Institute of Science, Bangalore\\
\texttt{rishikeshg@iisc.ac.in}
\and
Sutanu Gayen\\
National University of Singapore\\
\texttt{sutanugayen@gmail.com}
\and
Yuhao Wang\\
National University of Singapore\\
\texttt{yohanna.wang0924@gmail.com}
}
\date{}

\usepackage{amssymb,amsthm,amsmath}
\usepackage{bbm}
\usepackage{graphicx}
\usepackage{enumerate}
\usepackage{algorithm}
\usepackage{algpseudocode}
\usepackage{subfigure}

\usepackage[bottom]{footmisc}
\usepackage{comment} 

\usepackage[capitalise,nameinlink]{cleveref}
\usepackage{thmtools}
\usepackage{thm-restate}
\theoremstyle{plain}
\newtheorem{theorem}{Theorem}[section]
\newtheorem{lemma}[theorem]{Lemma}
\newtheorem{proposition}[theorem]{Proposition}

\newtheorem{fact}[theorem]{Fact}
\newtheorem{definition}[theorem]{Definition}
\newtheorem{corollary}[theorem]{Corollary}
\newtheorem{assumption}[theorem]{Assumption}
\newtheorem{claim}[theorem]{Claim}

\theoremstyle{remark}

\usepackage{mathtools}

\usepackage{tikz}
\usetikzlibrary{calc, graphs, graphs.standard, shapes, arrows, positioning, decorations.pathreplacing, decorations.markings, decorations.pathmorphing, fit, matrix, patterns, shapes.misc}

\newcommand{\eps}{\varepsilon}
\newcommand{\cC}{\mathcal{C}}
\newcommand{\cO}{\mathcal{O}}
\newcommand{\cP}{\mathcal{P}}
\newcommand{\cQ}{\mathcal{Q}}
\newcommand{\E}{\mathbb{E}}
\newcommand{\R}{\mathbb{R}}

\newcommand{\Norm}[1]{\left\lVert#1\right\rVert}
\newcommand{\Abs}[1]{\left\lvert#1\right\rvert}
\newcommand{\Paren}[1]{\left(#1\right)}
\newcommand{\Brac}[1]{\left[#1\right]}
\newcommand{\Brace}[1]{\left\{#1\right\}}
\newcommand{\wh}{\widehat}
\newcommand{\wt}{\widetilde}
\newcommand{\cau}{\mathrm{Cauchy}}
\newcommand{\med}{\mathrm{median}}
\newcommand{\tv}{\mathrm{d_{TV}}}
\newcommand{\kl}{\mathrm{d_{KL}}}
\newcommand{\cp}{\mathrm{d_{CP}}}
\newcommand{\gauss}{N}

\newcommand{\ignore}[1]{}

\newcommand*\justify{%
  \fontdimen2\font=0.4em
  \fontdimen3\font=0.2em
  \fontdimen4\font=0.1em
  \fontdimen7\font=0.1em
  \hyphenchar\font=`\-
}

\usepackage[textwidth=1.9cm]{todonotes}
\newcommand{\arnab}[1]{\todo[size=\tiny, shadow, color=green!40]{AB: #1}}
\newcommand{\sutanu}[1]{\todo[size=\tiny, shadow, color=blue!20]{SG: #1}}
\newcommand{\davin}[1]{\todo[size=\tiny, shadow, color=pink!40]{DC: #1}}
\newcommand{\rishi}[1]{\todo[size=\tiny, shadow, color=yellow!40]{RG: #1}}
\newcommand{\yuhao}[1]{\todo[size=\tiny, shadow, color=orange!40]{YW: #1}}

\renewcommand{\arnab}[1]{}\renewcommand{\sutanu}[1]{}\renewcommand{\davin}[1]{}\renewcommand{\rishi}[1]{}\renewcommand{\yuhao}[1]{}\renewcommand{\todo}[1]{}

\begin{document}

\maketitle

\begin{abstract}
Gaussian Bayesian networks (a.k.a.\ linear Gaussian structural equation models) are widely used to model causal interactions among continuous variables.
In this work, we study the problem of learning a fixed-structure Gaussian Bayesian network up to a bounded error in total variation distance.
We analyze the commonly used node-wise least squares regression \texttt{LeastSquares} and prove that it has the near-optimal sample complexity.
We also study a couple of new algorithms for the problem:
\begin{itemize}
\item \texttt{BatchAvgLeastSquares} takes the average of several batches of least squares solutions at each node, so that one can interpolate between the batch size and the number of batches. We show that \texttt{BatchAvgLeastSquares} also has near-optimal sample complexity. 
\item \texttt{CauchyEst} takes the median of solutions to several batches of linear systems at each node. We show that the algorithm specialized  to polytrees, \texttt{CauchyEstTree}, has near-optimal sample complexity.
\end{itemize}

Experimentally\footnote{Code is available at \url{https://github.com/YohannaWANG/CauchyEst}}, we show that for uncontaminated, realizable data, the \texttt{LeastSquares} algorithm performs best, but in the presence of contamination or DAG misspecification, \texttt{CauchyEst}/\texttt{CauchyEstTree} and \texttt{BatchAvgLeastSquares} respectively perform better.

\end{abstract}

\section{Introduction}
\label{sec:intro}

Linear structural equation models (SEMs) are widely used to model interactions among multiple variables, and they have found applications in diverse areas, including economics, genetics, sociology, psychology and education; see~\cite{Mueller99, Mulaik09} for pointers to the extensive literature in this area.
Gaussian Bayesian networks are a popularly used form of SEMs where:
(i) there are no hidden/unobserved variables,
(ii) each variable is a linear combination of its parents plus a Gaussian noise term, and
(iii) all noise terms are independent.
The \emph{structure} of a Bayesian network refers to the directed acyclic graph (DAG) that prescribes the parent nodes for each node in the graph.

This work studies the task of learning a Gaussian Bayesian network given its structure.
The problem is obviously quite fundamental and has been subject to extensive prior work.
The usual formulation of this problem is in terms of \emph{parameter estimation}, where one wants a consistent estimator that \emph{exactly} recovers the parameters of the Bayesian network in the limit, as the the number of samples approaches infinity.
In contrast, we consider the problem from the viewpoint of \emph{distribution learning}~\cite{KMRRSS94}.
Analogous to Valiant's famous PAC model for learning Boolean functions~\cite{Val84}, the goal here is to learn, with high probability, a distribution $\wh{\cP}$ that is close to the ground-truth distribution $\cP$, using an efficient algorithm.
In this setting, pointwise convergence of the parameters is no longer a requirement; the aim is rather to approximately learn the induced distribution.
Indeed, this relaxed objective may be achievable when the former may not be (e.g., for ill-conditioned systems) and can be the more relevant requirement for downstream inference tasks.
Diakonikolas~\cite{Diakonikolas16} surveys the current state-of-the-art in distribution learning from an algorithmic perspective. 

Consider a Gaussian Bayesian network $\cP$ with $n$ variables with parameters in the form of coefficients\footnote{We write $i \leftarrow j$ to emphasize that $X_j$ is the parent of $X_i$ in the Bayesian network.} $a_{i \leftarrow j}$'s and variances $\sigma_i$'s.
For each $i \in [n]$, the \emph{linear structural equation} for variable $X_i$, with parent indices $\pi_i \subseteq [n]$, is as follows:
\[
X_i
= \eta_i + \sum_{j \in \pi_i} a_{i \leftarrow j} X_j
, \qquad \eta_i \sim N(0, \sigma_i^2)
\]
for some unknown parameters $a_{i \leftarrow j}$'s and $\sigma_i$'s.
If a variable $X_i$ has no parents, then $X_i \sim N(0, \sigma_i^2)$ is simply an independent Gaussian.
Our goal is to efficiently learn a distribution $\cQ$ such that
\[
\tv(\cP,\cQ)\leq \eps
\]
while minimizing the number of samples we draw from $\cP$ as a function of $n$, $\eps$ and relevant structure parameters.
Here, $\tv$ denotes the total variation or statistical distance between two distributions:
\[
\tv(\cP, \cQ) = \frac{1}{2} \int_{\R^n} \Abs{\cP(x)-\cQ(x)} dx.
\]

One way to approach the problem is to simply view $\cP$ as an $n$-dimensional multivariate Gaussian.
In this case, it is known that the Gaussian $\cQ \sim N(0,\wh{\Sigma})$ defined by the empirical covariance matrix $\wh{\Sigma}$, computed with $\cO(n^2/\eps^2)$ samples from $\cP$, is $\eps$-close in TV distance to $\cP$ with constant probability~\cite[Appendix C]{ashtiani2020near}.
This sample complexity is also necessary for learning general $n$-dimensional Gaussians and hence general Gaussian Bayesian networks on $n$ variables.

We focus therefore on the setting where the structure of the network is \emph{sparse}.
\begin{assumption}
\label{assumption:d-parents}
Each variable in the Bayesian network has at most $d$ parents.
i.e.\ $\Abs{\pi_i} \leq d, \forall i \in [n]$.
\end{assumption}
Sparsity is a common and very useful assumption for statistical learning problems; see the book~\cite{HTW19} for an overview of the role of sparsity in regression. 
More specifically, in our context, the assumption of bounded in-degree is popular (e.g., see~\cite{Dasgupta97, BCD20}) and also very natural, as it reflects the belief that in the correct model, each linear structural equation involves a small number of variables\footnote{More generally, one would want to assume a bound on the ``complexity'' of each equation. In our context, as each equation is linear, the complexity is simply proportional to the in-degree.}.

\subsection{Our contributions}

\begin{enumerate}
    \item Analysis of MLE \texttt{LeastSquares} and a distributed-friendly generalization \texttt{BatchAvgLeastSquares}\\
    \hspace{0pt}\\
    The standard algorithm for parameter estimation of Gaussian Bayesian networks is to perform node-wise least squares regression.
    It is easy to see that \texttt{LeastSquares} is the maximum likelihood estimator.
    However, to the best of our knowledge, there did not exist an explicit sample complexity bound for this estimator.
    We show that the sample complexity for learning $\cP$ to within TV distance $\eps$ using \texttt{LeastSquares} requires only $\wt{\cO}(nd/\eps^2)$ samples, which is nearly optimal.
    
    We also give a generalization dubbed \texttt{BatchAvgLeastSquares} which solves multiple batches of least squares problems (on smaller systems of equations), and then returns their mean.
    As each batch is independent from the others, they can be solved separately before their solutions are combined.
    Our analysis will later show that we can essentially interpolate between ``batch size'' and ``number of batches'' while maintaining a similar total sample complexity -- \texttt{LeastSquares} is the special case of a single batch.
    Notably, we do not require any additional assumptions about the coefficients or variance terms in the analyses of \texttt{LeastSquares} and \texttt{BatchAvgLeastSquares}.

    \item New algorithms based on Cauchy random variables: \texttt{CauchyEst} and \texttt{CauchyEstTree}\\
    \hspace{0pt}\\
    We develop a new algorithm \texttt{CauchyEst}.
    At each node, \texttt{CauchyEst} solves several small linear systems of equations and takes the component-wise median of the obtained solution to obtain an estimate of the coefficients for the corresponding structural equation.
    In the special case of bounded-degree \emph{polytree} structures, where the underlying undirected Bayesian network is acyclic, we specialize the algorithm to \texttt{CauchyEstTree} for which we give theoretical guarantees.
    Polytrees are of particular interest because inference on polytree-structured Bayesian networks can be performed efficiently~\cite{KP83, pearl1986fusion}.
    On polytrees, we show that the sample complexity of \texttt{CauchyEstTree} is also $\wt{\cO}(nd/\eps^2)$.
    Somewhat surprisingly, our analysis (as the name of the algorithm reflects) involves Cauchy random variables which usually do not arise in the context of regression.

    \item Hardness results\\
    \hspace{0pt}\\
    In \cref{sec:hardness}, we show that our sample complexity upper-bound is nearly optimal in terms of the dependence on the parameters $n$, $d$, and $\epsilon$.
    In particular, we show that learning the coefficients and noises of a linear structural equation model (on $n$ variables, with in-degree at most $d$) up to an $\epsilon$-error in $\tv$-distance with probability 2/3 requires $\Omega(nd\epsilon^{-2})$ samples in general.
    We use a packing argument based on Fano's inequality to achieve this result.
    
    \item Extensive experiments on synthetic and real-world networks \\
    \hspace{0pt}\\
    We experimentally compare the algorithms studied here as well as other well-known methods for undirected Gaussian graphical model regression, and investigate how the distance between the true and learned distributions changes with the number of samples. We find that the \texttt{LeastSquares} estimator performs the best among all algorithms on uncontaminated datasets. However, \texttt{CauchyEst}, \texttt{CauchyEstTree} and \texttt{BatchMedLeastSquares} outperform \texttt{LeastSquares} and \texttt{BatchAvgLeastSquares} by a large margin when a fraction of the samples are contaminated. In non-realizable/agnostic learning case, \texttt{BatchAvgLeastSquares}, \texttt{CauchyEst}, and \texttt{CauchyEstTree} performs better than the other algorithms.
\end{enumerate}

\subsection{Outline of paper}

In \cref{sec:preliminaries}, we relate KL divergence with TV distance and explain how to decompose the KL divergence into $n$ terms so that it suffices for us to estimate the parameters for each variable independently.
We also give an overview of our two-phased recovery approach and explain how to use recovered coefficients to estimate the variances via \texttt{VarianceRecovery}.
For estimating coefficients, \cref{sec:least-squares} covers the estimators based on linear least squares (\texttt{LeastSquares} and \texttt{BatchAvgLeastSquares}) while \cref{sec:cauchyest} presents our new Cauchy-based algorithms (\texttt{CauchyEst} and \texttt{CauchyEstTree}).
To complement our algorithmic results, we provide hardness results in \cref{sec:hardness} and experimental evaluation in \cref{sec:experiments}.

For a cleaner exposition, we defer some formal proofs to \cref{sec:deferred-proofs}.

\subsection{Further related work}

Bayesian networks, both in discrete and continuous settings, were formally introduced by Pearl~\cite{pearl} in 1988 to model uncertainty in AI systems.
For the continuous case, Pearl considered the case when each node is a linear function of its parents added with an independent Gaussian noise~\cite[Chapter 7]{pearl}.
The parameter learning problem -- recovering the distribution of nodes conditioned on its parents from data -- is well-studied in practice, and maximum likelihood estimators are known for various simple settings such as when the conditional distribution is Gaussian or the variables are discrete-valued.
See for example the implementation of \texttt{fit} in the R package \texttt{bnlearn}~\cite{bnlearn}.

The focus of our paper is to give formal guarantees for the parameter learning in the PAC framework introduced by Valiant~\cite{Val84} in 1984.
Subsequently, Haussler~\cite{haussler} generalized this framework for studying parameter and density estimation problems of continuous distributions.
Dasgupta~\cite{Dasgupta97} first looked at the problem of parameter learning for fixed structure Bayesian networks in the discrete and continuous settings and gave finite sample complexity bounds for these problems based on the VC-dimensions of the hypothesis classes.
In particular, he gave an algorithm for learning the parameters of a Bayes net on $n$ binary variables of bounded in-degree in $\kl$ distance using a quadratic in $n$ samples.
Subsequently, tight (linear) sample complexity upper and lower bounds were shown for this problem~\cite{BGMV,bgpv,CDKS}.
To the best of our knowledge, a finite PAC-style bound for fixed-structure Gaussian Bayesian networks was not known previously.

The question of structure learning for Gaussian Bayesian networks has been extensively studied.
A number of works~\cite{PB14, GH17, chen2019causal, PK20, park2020identifiability, GDA20} have proposed increasingly general conditions for ensuring identifiability of the network structure from observations.
Structure learning algorithms that work for high-dimensional Gaussian Bayesian networks have also been proposed by others (e.g., see~\cite{AQ15, AGZ19, GZ20}).

\section{Preliminaries}
\label{sec:preliminaries}

In this section, we discuss why we upper bound the total variational distance using KL divergence and give a decomposition of the KL divergence into $n$ terms, one associated with each variable in the Bayesian network.
This decomposition motivates why our algorithms and analysis focus on recovering parameters for a single variable.
We also present our general two-phased recovery approach and explain how to estimate variances using recovered coefficients in \cref{sec:two-phased}.

\subsection{Notation}
\label{sec:notation}

A Bayesian network (Bayes net in short) $\cP$ is a joint distribution $\cP$ over $n$ variables $X_1, \ldots, X_n$ defined by the underlying directed acyclic graph (DAG) $G$.
The DAG $G = (V,E)$ encodes the dependence between the variables where $V = \Brace{X_1, \ldots, X_n}$ and $(X_j, X_i) \in E$ if and only if $X_j$ is a parent of $X_i$.
For any variable $X_i$, we use the set $\pi_i \subseteq [n]$ to represent the indices of $X_i$'s parents.
Under this notation, each variable $X_i$ of $\cP$ is independent of $X_i$'s non-descendants conditioned on $\pi_i$.
Therefore, using Bayes rule in the topological order of $G$, we have
\[
\cP(X_1, \ldots, X_n) = \prod_{i=1}^n \Pr_{\cP}(X_i \mid \pi_i)
\]

Without loss of generality, by renaming variables, we may assume that each variable $X_i$ only has ancestors with indices smaller than $i$.
We also define $p_i = \Abs{\pi_i}$ as the number of parents of $X_i$ and $d_{avg} = \frac{1}{n} \sum_{i=1}^n p_i$ to be average in-degree.
Furthermore, a DAG $G$ is a \emph{polytree} if the undirected version of $G$ is a acyclic.

We study the \emph{realizable} setting where our unknown probability distribution $\cP$ is Markov with respect to the given Bayesian network.
We denote the true (hidden) parameters associated with $\cP$ by $\alpha^* = (\alpha^*_1, \ldots , \alpha^*_n)$.
Our algorithms recover parameter estimates $\wh{\alpha} = (\wh{\alpha}_1, \ldots , \wh{\alpha}_n)$ such that the induced probability distribution $\cQ$ is close in total variational distance to $\cP$.
For each $i \in [n]$, $\alpha_i^* = (A_i, \sigma_i)$ is the set of ground truth parameters associated with variable $X_i$, $A_i$ is the coefficients associated to $\pi_i$, $\sigma_i^2$ is the variance of $\eta_i$, $\wh{\alpha}_i^* = (\wh{A}_i, \wh{\sigma}_i)$ is our estimate of $\alpha_i^*$.

In the course of the paper, we will often focus on the perspective a single variable of interest.
This allows us to drop a subscript for a cleaner discussion.
Let us denote such a variable of interest by $Y \in V$ and use the index $y \in [n]$.
Without loss of generality, by renaming variables, we may further assume that the parents of $Y$ are $X_1, \ldots, X_p$.
By \cref{assumption:d-parents}, we know that $p \leq d$.
We can write $Y = \eta_y + \sum_{i=1}^p a_{y \leftarrow i} X_i$ where $\eta_y \sim N(0, \sigma_y^2)$.
We use matrix $M \in \R^{p \times p}$ to denote the covariance matrix defined by the parents of $Y$, where $M_{i,j} = \E\Brac{X_i X_j}$ and $M = LL^\top$ is the Cholesky decomposition of $M$.
Under this notation, we see the vector $(X_1, \ldots, X_p) \sim N(0, M)$ is distributed as a multivariate Gaussian.
Our goal is then to produce estimates $\wh{a}_{y \leftarrow i}$ for each $a_{y \leftarrow i}$.
For notational convenience, we can group the coefficients into $A = \Brac{a_{y \leftarrow 1}, \ldots, a_{y \leftarrow p}}^\top$ and $\wh{A} = \Brac{\wh{a}_{y \leftarrow 1}, \ldots, \wh{a}_{y \leftarrow p}}^\top$.
The vector $\Delta = (\wh{A} - A)^\top$ captures the entry-wise gap between our estimates and the ground truth.

We write $[n]$ to mean $\{1, 2, \ldots, n\}$ and $\Abs{S}$ to mean the size of a set $S$.
For a matrix $M$, $M_{i,j}$ denotes its $(i,j)$-th entry.
We use $\Norm{\cdot}$ to both mean the operator/spectral norm for matrices and $L_2$-norm for vectors, which should be clear from context.
We hide absolute constant multiplicative factors and multiplicative factors poly-logarithmic in the argument using standard notations: $\cO(\cdot)$, $\Omega(\cdot)$, and $\widetilde{\cO}(\cdot)$.

\subsection{Basic facts and results}

We begin by stating some standard facts and results.

\begin{fact}
\label{fact:standard-gaussian-tail-bound}
Suppose $X \sim N(0, \sigma^2)$.
Then, for any $t > 0$,
$
\Pr\Paren{X > t} \leq \exp\Paren{-\frac{t^2}{2 \sigma^2}}
$.
\end{fact}

\begin{fact}
\label{fact:linear-combination-of-gaussians}
Consider any matrix $B \in \R^{n \times m}$ with rows $B_1, \ldots, B_n \in \R^m$.
For any $i \in [m]$ and any vector $v \in \R^m$ with i.i.d.\ $N(0, \sigma^2)$ entries, we have that
$
(Bv)_i = B_i v \sim N(0, \sigma^2 \cdot \Norm{B_i}^2)
$.
\end{fact}

\begin{fact}[Theorem 2.2 in \cite{Gut_2009})]
\label{fact:multivariate-transformation}
Let $X_1, \ldots, X_p \sim N(0, LL^\top)$ be $p$ i.i.d.\ $n$-dimensional multivariate Gaussians with covariance $LL^\top \in \R^{n \times n}$ (i.e.\ $L \in \R^{n \times p}$).
If $X \in \R^{p \times n}$ is the matrix formed by stacking $X_1, \ldots, X_p$ as rows of $X$, then $X = GL^\top$ where $G \in \R^{p \times p}$ is a random matrix with i.i.d.\ $N(0,1)$ entries.\footnote{The transformation stated in \cite[Theorem 2.2, page 120]{Gut_2009} is for a single multivariate Gaussian vector, thus we need to take the transpose when we stack them in rows. Note that $G$ and $G^\top$ are identically distributed.}
\end{fact}

\begin{lemma}[Equation 2.19 in \cite{wainwright2019high}]
\label{lem:chi-square-concentration}
Let $Y = \sum_{k=1}^n Z_k^2$, where each $Z_k \sim N(0,1)$.
Then, $Y \sim \chi^2_n$ and for any $0 < t < 1$, we have
$
\Pr\Paren{ \Abs{Y/n - 1} \geq t }
\leq 2 \exp\Paren{ -nt^2/8 }
$.
\end{lemma}

\begin{lemma}[Consequence of Corollary 3.11 in \cite{bandeira2016sharp}]
\label{sec:norm-of-rectangular-gaussian-matrix}
Let $G \in \R^{n \times m}$ be a matrix with i.i.d.\ $N(0,1)$ entries where $n \leq m$.
Then, for some universal constant $C$,
$
\Pr\Paren{\Norm{G} \geq 2(\sqrt{n} + \sqrt{m})} \leq \sqrt{n} \cdot \exp\Paren{-C \cdot m}
$.
\end{lemma}

\begin{restatable}{lemma}{gtg}
\label{lem:gtg}
Let $G \in \R^{k \times d}$ be a matrix with i.i.d.\ $N(0,1)$ entries.
Then, for any constant $0 < c_1 < 1/2$ and $k \geq d / c_1^2$,
\[
\Pr\Paren{ \Norm{(G^\top G)^{-1}}_{op} \leq \frac{1}{\Paren{1 - 2 c_1}^2 k} } \geq 1 - \exp\Paren{- \frac{k c_1^2}{2}}
\]
\end{restatable}
\begin{proof}
See \cref{sec:deferred-proofs}.
\end{proof}

\begin{restatable}{lemma}{Getanorm}
\label{lem:Getanorm}
Let $G \in \R^{k \times p}$ be a matrix with i.i.d.\ $N(0,1)$ entries and $\eta \in \R^{k}$ be a vector with i.i.d.\ $N(0, \sigma^2)$ entries, where $G$ and $\eta$ are independent.
Then, for any constant $c_2 > 0$,
\[
\Pr\Paren{ \Norm{G^\top \eta}_2 < 2 \sigma c_2 \sqrt{k p}}
\geq 1 - 2p \exp\Paren{-2k} - p \exp\Paren{-\frac{c_2^2}{2}}
\]
\end{restatable}
\begin{proof}
See \cref{sec:deferred-proofs}.
\end{proof}

The next result gives the non-asymptotic convergence of medians of Cauchy random variables.
We use this result in the analysis of \texttt{CauchyEst}, and it may be of independent interest.

\begin{restatable}{lemma}{cauchymedianconvergence}[Non-asymptotic convergence of Cauchy median]
\label{lem:cauchy-median-convergence}
Consider a collection of $m$ i.i.d.\ $\cau(0,1)$ random variables $X_1, \ldots, X_m$.
Given a threshold $0 < \tau < 1$, we have
\[
\Pr\Paren{\med\Brace{X_1, \ldots, X_m} \not\in [-\tau, \tau]} \leq 2 \exp\Paren{-\frac{m \tau^2}{8}}
\]
\end{restatable}
\begin{proof}
See \cref{sec:deferred-proofs}.
\end{proof}

\subsection{Distance and divergence between probability distributions}
\label{sec:KL-decomposition}

Recall that we are given sample access to an unknown probability distribution $\cP$ over the values of $(X_1, \ldots, X_n) \in \R^n$ and the corresponding structure of a Bayesian network $G$ on $n$ variables.
We denote $\alpha^*$ as the true (hidden) parameters, inducing the probability distribution $\cP$, which we estimate by $\wh{\alpha}$.
In this work, we aim to recover parameters $\wh{\alpha}$ such that our induced probability distribution $\cQ$ is as close as possible to $\cP$ in total variational distance.

\begin{definition}[Total variational (TV) distance]
Given two probability distributions $\cP$ and $\cQ$ over $\R^n$, the total variational distance between them is defined as
$
\tv(\cP, \cQ)
= \sup_{A \in \R^n} \Abs{\cP(A) - \cQ(A)}
= \frac{1}{2} \int_{\R^n} \Abs{\cP - \cQ} dx
$.
\end{definition}

Instead of directly dealing with total variational distance, we will instead bound the Kullback–Leibler (KL) divergence and then appeal to the Pinsker's inequality \cite[Lemma 2.5, page 88]{tsybakov2008introduction} to upper bound $\tv$ via $\kl$.
We will later show algorithms that achieve $\kl(\cP, \cQ) \leq \eps$.

\begin{definition}[Kullback–Leibler (KL) divergence]
Given two probability distributions $\cP$ and $\cQ$ over $\R^n$, the KL divergence between them is defined as
$
\kl(\cP, \cQ) = \int_{A \in \R^n} \cP(A) \log \Paren{ \frac{\cP(A)}{\cQ(A)} } dA
$.
\end{definition}

\begin{fact}[Pinsker's inequality]
For distributions $\cP$ and $\cQ$,
$
\tv(\cP, \cQ) \leq \sqrt{\kl(\cP, \cQ) / 2}
$.
\end{fact}

Thus, if $s(\eps)$ samples are needed to learn a distribution $\cQ$ such that $\kl(\cP, \cQ) \leq \eps$, $s(\eps^2)$ samples are needed to ensure $\tv(\cP, \cQ)\leq \eps$.

\subsection{Decomposing the KL divergence}

For a set of parameters $\alpha = (\alpha_1, \ldots, \alpha_n)$, denote $\alpha_i$ as the subset of parameters that are relevant to the variable $X_i$.
Following the approach of \cite{Dasgupta97}\footnote{\cite{Dasgupta97} analyzes the \emph{non-realizable} setting where the distribution $\cP$ may not correspond to the causal structure of the given Bayesian network. As we study the \emph{realizable} setting, we have a much simpler derivation.}, we decompose $\kl(\cP, \cQ)$ into $n$ terms that can be computed by analyzing the quality of recovered parameters for each variable $X_i$.

For notational convenience, we write $x$ to mean $(x_1, \ldots, x_n)$, $\pi_i(x)$ to mean the values given to parents of variable $X_i$ by $x$, and $\cP(x)$ to mean $\cP(X_1 = x_1, \ldots, X_n = x_n)$.
Let us define
\[
\cp(\alpha^*_i, \wh{\alpha}_i)
= \int_{x_i, \pi_i(x)} \cP(x_i, \pi_i(x)) \log
\Paren{\frac{\cP(x_i \mid \pi_i(x))}{\cQ(x_i \mid \pi_i(x))}} d x_i d \pi_i(x)    
\]
where each $\wh{\alpha}_i$ and $\alpha^*_i$ represent the parameters that relevant to variable $X_i$ from $\wh{\alpha}$ and $\alpha^*$ respectively.
By the Bayesian network decomposition of joint probabilities and marginalization, one can show that 
\[
\kl(\cP, \cQ)
= \sum_{i=1}^n \cp(\alpha^*_i, \wh{\alpha}_i)
\]
See \cref{sec:decomposition-details} for the full derivation details.

\subsection{Bounding $\cp$ for an arbitrary variable}

We now analyze $\cp(\alpha^*_i, \wh{\alpha}_i)$ with respect to the our estimates $\wh{\alpha}_i = (\wh{A}_i, \wh{\sigma}_i)$ and the hidden true parameters $\alpha^*_i = (A_i, \sigma_i)$ for any $i \in [n]$.
For derivation details, see \cref{sec:decomposition-details}.

With respect to variable $X_i$, one can derive
$
\cp(\alpha^*_i, \wh{\alpha}_i)
= \ln \Paren{\frac{\wh{\sigma}_i}{\sigma_i}} + \frac{\sigma_i^2 - \wh{\sigma}_i^2}{2 \wh{\sigma}_i^2} + \frac{\Delta_i^\top M_i \Delta_i}{2 \wh{\sigma}_i^2}
$.
Thus,
\begin{equation}
\label{eq:decomposition}
\kl(\cP, \cQ)
= \sum_{i=1}^n \cp(\alpha^*_i, \wh{\alpha}_i)
= \sum_{i=1}^n \ln \Paren{\frac{\wh{\sigma}_i}{\sigma_i}} + \frac{\sigma_i^2 - \wh{\sigma}_i^2}{2 \wh{\sigma}_i^2} + \frac{\Delta_i^\top M_i \Delta_i}{2 \wh{\sigma}_i^2}
\end{equation}
where $M_i$ is the covariance matrix associated with variable $X_i$, $\alpha_i^* = (A_i, \sigma_i)$ is the coefficients and variance associated with variable $X_i$, $\alpha_i = (\wh{A}_i, \wh{\sigma}_i)$ are the estimates for $\alpha_i^*$, and $\Delta_i = \wh{A}_i - A_i$.

\begin{proposition}[Implication of KL decomposition]
\label{prop:decomposition-implication}
Let $\eps \leq 0.17$ be a constant.
Suppose $\wh{\alpha}_i$ has the following properties for each $i \in [n]$:
\begin{equation}
\label{eq:condition1}
\tag{Condition 1}
\Abs{\Delta_i^\top M_i \Delta_i} \leq \sigma_i^2 \cdot \frac{\eps \cdot p_i}{n \cdot d_{avg}}
\end{equation}
\begin{equation}
\label{eq:condition2}
\tag{Condition 2}
\Paren{1 - \sqrt{\frac{\eps \cdot p_i}{n \cdot d_{avg}}}} \cdot \sigma_i^2 \leq \wh{\sigma}_i^2 \leq \Paren{1 + \sqrt{\frac{\eps \cdot p_i}{n \cdot d_{avg}}}} \cdot \sigma_i^2
\end{equation}
Then, $\cp(\alpha^*_i, \wh{\alpha}_i) \leq 3 \cdot \frac{\eps \cdot p_i}{n \cdot d_{avg}}$ for all $i \in [n]$.
Thus, $\kl(\cP, \cQ) = \sum_{i=1}^n \cp(\alpha^*_i, \wh{\alpha}_i) \leq 3 \eps$.\footnote{For a cleaner argument, we are bounding $\kl(\cP, \cQ) \leq 3\eps$. This is qualitatively the same as showing $\kl(\cP, \cQ) \leq \eps$ since one can repeat the entire analysis with $\eps' = \eps/3$.}
\end{proposition}
\begin{proof}
Consider an arbitrary fixed $i \in [n]$. Denote $\gamma = \sigma_i^2 / \wh{\sigma}_i^2$.
Observe\footnote{This inequality is also used in \cite[Lemma 2.9]{ashtiani2020near}.} that $\gamma - 1 - \ln(\gamma) \leq (\gamma-1)^2$ for $\gamma \geq 0.316\ldots$.
Since $p_i \leq n \cdot d_{avg} = \sum_i p_i$, \cref{eq:condition2} implies that $\gamma \geq 1/(1+\sqrt{\eps}) \geq 1/2$.
Then,
\begin{align*}
\ln \Paren{\frac{\wh{\sigma}_i}{\sigma_i}} + \frac{\sigma_i^2 - \wh{\sigma}_i^2}{2 \wh{\sigma}_i^2}
&= \frac{1}{2} \cdot \Paren{ \frac{\sigma_i^2}{\wh{\sigma}_i^2} - 1 - \ln \Paren{\frac{\sigma_i^2}{\wh{\sigma}_i^2}}}\\
&= \frac{1}{2} \cdot \Paren{ \gamma - 1 - \ln \Paren{\gamma}}\\
&\leq \frac{1}{2} \cdot \Paren{\gamma - 1}^2 && \text{By \ref{eq:condition2}}\\
&\leq \frac{1}{2} \cdot \Paren{\frac{1}{1 - \sqrt{\frac{\epsilon p_i}{n d_{avg}}}} - 1}^2 && \text{Since $\Paren{1 - \sqrt{\frac{\epsilon p_i}{n d_{avg}}}} \cdot \sigma_i^2 \leq \wh{\sigma}_i^2$}\\
&\leq \frac{2 \epsilon p_i}{n d_{avg}} && \text{Holds when $0 \leq \frac{\epsilon p_i}{n d_{avg}} \leq \frac{1}{4}$}
\end{align*}

Meanwhile,
\begin{align*}
\frac{\Delta_i^\top M_i \Delta_i}{2 \wh{\sigma}_i^2}
&\leq \frac{\Abs{\Delta_i^\top M_i \Delta_i}}{2 \wh{\sigma}_i^2}\\
&\leq \frac{p_i \epsilon}{n d_{avg}} \cdot \frac{\sigma_i^2}{2 \wh{\sigma}_i^2} && \text{By \ref{eq:condition1}}\\
&\leq \frac{p_i \epsilon}{2n d_{avg}} \cdot \frac{1}{1 - \sqrt{\frac{\epsilon p_i}{n d_{avg}}}} && \text{Since $\Paren{1 - \sqrt{\frac{\epsilon p_i}{n d_{avg}}}} \cdot \sigma_i^2 \leq \wh{\sigma}_i^2$}\\
&\leq \frac{\epsilon p_i}{n d_{avg}} && \text{Holds when $0 \leq \frac{\epsilon p_i}{n d_{avg}} \leq \frac{1}{4}$}
\end{align*}

Putting together, we see that $d_{CP}(\alpha^*_i, \wh{\alpha}_i) \leq \frac{3 \epsilon p_i}{n d_{avg}}$.
\end{proof}

\subsection{Two-phased recovery approach}
\label{sec:two-phased}

\cref{alg:recovery-two-phased} states our two-phased recovery approach.
We estimate the coefficients of the Bayesian network in the first phase and use them to recover the variances in the second phase.

\begin{algorithm}[htbp]
\caption{Two-phased recovery algorithm}
\label{alg:recovery-two-phased}
\begin{algorithmic}[1]
    \State \textbf{Input}: DAG $G$ and sample parameters $m_1$ and $m_2$
    \State Draw $m = m_1 + m_2$ independent samples of $(X_1, \ldots, X_n)$.
    \State $\wh{A}_1, \ldots, \wh{A}_n \leftarrow$ Run a coefficient recovery algorithm using first $m_1$ samples.
    \State $\wh{\sigma}_1^2, \ldots, \wh{\sigma}_n^2 \leftarrow$ Run \texttt{VarianceRecovery} using last $m_2$ samples and $\wh{A}_1, \ldots, \wh{A}_n$
    \State \Return $\wh{A}_1, \ldots, \wh{A}_n, \wh{\sigma}_1^2, \ldots, \wh{\sigma}_n^2$
\end{algorithmic}
\end{algorithm}

Motivated by \cref{prop:decomposition-implication}, we will estimate parameters for each variable in an independent fashion\footnote{Given the samples, parameters related to each variable can be estimated in parallel.}.
We will provide various coefficient recovery algorithms in the subsequent sections.
These algorithms will recover coefficients $\wh{A}_i$ that satisfy \ref{eq:condition1} for each variable $X_i$.
We evaluate them empirically in \cref{sec:experiments}.
For variance recovery, we use \texttt{VarianceRecovery} for each variable $Y$ by computing the empirical variance\footnote{Except in our experiments with contaminated data in Section~\ref{sec:experiments} where we use the classical median absolute devation (MAD) estimator. See \cref{sec:mad} for a description.} of $Y - X \wh{A}$ such that the recovered variance satisfies \ref{eq:condition2}.

\begin{algorithm}[htbp]
\caption{\texttt{VarianceRecovery}: Variance recovery algorithm given coefficient estimates}
\label{alg:recovery-variance}
\begin{algorithmic}[1]
    \State \textbf{Input}: DAG $G$, coefficient estimates, and $m_2 \in \cO\Paren{\frac{n d_{avg}}{\eps} \log\Paren{\frac{n}{\delta}}}$ samples
    \For{variable $Y$ with $p$ parents and coefficient estimate $\wh{A}$} \Comment{If $p=0$, then $\wh{A} = 0$.}
        \State Without loss of generality, by renaming variables, let $X_1, \ldots, X_p$ be the parents of $Y$.
        \For{$s = 1, \ldots, m_2$}
            \State Define $Y^{(s)}$ as the $s^{th}$ sample of $Y$.
            \State Define $X^{(s)} = [X^{(s)}_1, \ldots, X^{(s)}_p]$ as the $s^{th}$ sample of $X_1, \ldots, X_p$ placed in a \emph{row} vector.
            \State Define $Z^{(s)} = \Paren{Y^{(s)} - X^{(s)} \wh{A}}^2$.
        \EndFor
        \State Estimate $\wh{\sigma}_y^2 = \frac{1}{m_2} \sum_{i=1}^{m_2} Z^{(s)}$
    \EndFor
\end{algorithmic}
\end{algorithm}

To analyze \texttt{VarianceRecovery}, we first prove guarantees for an arbitrary variable and then take union bound over $n$ variables.

\begin{lemma}
\label{lem:recovery-variance-single}
Consider \cref{alg:recovery-variance}.
Fix any arbitrary variable of interest $Y$ with $p$ parents, parameters $(A, \sigma_y)$, and associated covariance matrix $M$.
Suppose we have coefficient estimates $\wh{A}$ such that $\Abs{\Delta^\top M \Delta} \leq \sigma_y^2 \cdot \frac{p \eps}{n d_{avg}}$.
Suppose $0 \leq \eps \leq 3 - 2\sqrt{2} \leq 0.17$.
With $k = \frac{32n d_{avg}}{\eps p} \log\Paren{\frac{2}{\delta}}$ samples, we recover variance estimate $\wh{\sigma}_y$ such that
\[
\Pr\Paren{\Paren{1 - \sqrt{\frac{\eps p}{n d_{avg}}}} \cdot \sigma_y^2 \leq \wh{\sigma}_y^2 \leq \Paren{1 + \sqrt{\frac{\eps p}{n d_{avg}}}} \cdot \sigma_y^2} \geq 1 - \delta
\]
\end{lemma}
\begin{proof}
We first argue that $\wh{\sigma}_y^2 \sim \Paren{\sigma_y^2 + \Delta^\top M \Delta} \cdot \chi^2_{k}$, then apply standard concentration bounds for $\chi^2$ random variables (see \cref{lem:chi-square-concentration}).

For any sample $s \in [k]$, we see that
$
Y^{(s)} - X^{(s)} \wh{A}
= X^{(s)} A + \eta_y^{(s)} - X^{(s)} \wh{A}
= \eta_y^{(s)} - X^{(s)} \Delta
$,
where $\Delta = \wh{A} - A \in \R^p$ is an unknown constant vector (because we do not actually know $A$).
For fixed $\Delta$, we see that $X^{(s)} \Delta \sim N(0, \Delta^\top M \Delta)$.
Since $\eta_y^{(s)} \sim N(0, \sigma_y^2)$ and $X^{(s)}$ are independent, we have that $Y^{(s)} - X^{(s)} \wh{A} \sim N(0, \sigma_y^2 + \Delta^\top M \Delta)$.
So, for any sample $s \in [k]$, $Z^{(s)} = (Y^{(s)} - X^{(s)} \wh{A})^2 \sim \Paren{\sigma_y^2 + \Delta^\top M \Delta} \cdot \chi^2_{1}$.
Therefore, $\wh{\sigma}_y = \frac{1}{k} \sum_{s=1}^k Z^{(s)} \sim (\sigma_y^2 + \Delta^\top M \Delta)/k \cdot \chi^2_k$.
Let us define
\[
\gamma = \frac{\wh{\sigma}_y^2}{\sigma_y^2} \cdot \Paren{\frac{1}{1 + \frac{\Delta^\top M \Delta}{\sigma_y^2}}}
\sim \frac{\chi^2_k}{k}
\]

Since $p \leq n d_{avg}$, if $\eps \leq 3 - 2\sqrt{2}$, then $\frac{\eps p}{n d_{avg}} \leq 3 - 2\sqrt{2} \leq 3 + 2\sqrt{2}$.
We first make two observations:
\begin{enumerate}
    \item For $0 \leq \frac{\eps p}{n d_{avg}} \leq 3 - 2\sqrt{2}$,
    $
    \Paren{1 + \sqrt{\frac{\eps p}{n d_{avg}}}} \cdot \Paren{\frac{1}{1 + \frac{\Delta^\top M \Delta}{\sigma_y^2}}}
    \geq 1 + \sqrt{\frac{\eps p}{4n d_{avg}}}
    $.
    \item For $0 \leq \frac{\eps p}{n d_{avg}} \leq 3 + 2\sqrt{2}$,
    $
    \Paren{1 - \sqrt{\frac{\eps p}{n d_{avg}}}} \cdot \Paren{\frac{1}{1 + \frac{\Delta^\top M \Delta}{\sigma_y^2}}}
    \leq 1 - \sqrt{\frac{\eps p}{4n d_{avg}}}
    $.
\end{enumerate}

Using \cref{lem:chi-square-concentration} with the above discussion, we have
\begin{align*}
&\; \Pr\Paren{\frac{\wh{\sigma}_y^2}{\sigma_y^2} \geq 1 + \sqrt{\frac{\eps p}{n d_{avg}}} \text{\quad or \quad} \frac{\wh{\sigma}_y^2}{\sigma_y^2} \leq 1 - \sqrt{\frac{\eps p}{n d_{avg}}}}\\
= &\; \Pr\Paren{\gamma \geq \Paren{1 + \sqrt{\frac{\eps p}{n d_{avg}}}} \cdot \Paren{\frac{1}{1 + \frac{\Delta^\top M \Delta}{\sigma_y^2}}} \text{\quad or \quad} \gamma \leq \Paren{1 - \sqrt{\frac{\eps p}{n d_{avg}}}} \cdot \Paren{\frac{1}{1 + \frac{\Delta^\top M \Delta}{\sigma_y^2}}}}\\
\leq &\; \Pr\Paren{\gamma \geq 1 + \sqrt{\frac{\eps p}{4n d_{avg}}} \text{\quad or \quad} \gamma \leq 1 - \sqrt{\frac{\eps p}{4n d_{avg}}}}\\
= &\; \Pr\Paren{\Abs{\gamma-1} \geq \sqrt{\frac{\eps p}{4n d_{avg}}}}\\
\leq &\; 2 \exp\Paren{ -\frac{k \eps p}{32 n d_{avg}} }
\end{align*}

The claim follows by setting $k = \frac{32n d_{avg}}{\eps p} \log\Paren{\frac{2}{\delta}}$.
\end{proof}

\begin{corollary}[Guarantees of \texttt{VarianceRecovery}]
\label{thm:recovery-variance}
Consider \cref{alg:recovery-variance}.
Suppose $0 \leq \eps \leq 3 - 2\sqrt{2} \leq 0.17$ and we have coefficient estimates $\wh{A}_i$ such that $\Abs{\Delta_i^\top M_i \Delta_i} \leq \sigma_i^2 \cdot \frac{\eps p_i}{n d_{avg}}$ for all $i \in [n]$.
With $m_2 \in \cO\Paren{\frac{n d_{avg}}{\eps} \log\Paren{\frac{n}{\delta}}}$ samples, we recover variance estimate $\wh{\sigma}_i$ such that
\[
\Pr\Paren{\forall i \in [n], \Paren{1 - \sqrt{\frac{\eps p_i}{n d_{avg}}}} \cdot \sigma_i^2 \leq \wh{\sigma}_i^2 \leq \Paren{1 + \sqrt{\frac{\eps p_i}{n d_{avg}}}} \cdot \sigma_i^2} \geq 1 - \delta
\]
The total running time is $\cO\Paren{\frac{n^2 d_{avg}^2}{\eps} \log \Paren{\frac{1}{\delta}}}$.
\end{corollary}
\begin{proof}
For each $i \in [n]$, apply \cref{lem:recovery-variance-single} with $\delta' = \delta/n$ and $m_2 = \frac{32n d_{avg}}{\eps} \log\Paren{\frac{2}{\delta}} \geq \max_{i \in [n]} \frac{32n d_{avg}}{\eps p_i} \log\Paren{\frac{2}{\delta}}$, then take the union bound over all $n$ variables.

The computational complexity for a variable with $p$ parents is $\cO(m_2 \cdot p)$.
Since $\sum_{i=1}^n p_i = n d_{avg}$, the total runtime is $\cO(m_2 \cdot n \cdot d_{avg})$.
\end{proof}

In \cref{sec:hardness}, we show that the sample complexity is nearly optimal in terms of the dependence on $n$ and $\eps$.
We remark that we use \emph{one} batch of samples to use for all the nodes; this is possible as we can obtain high-probability bounds on the error events at each node.

\section{Coefficient estimators based on linear least squares}
\label{sec:least-squares}

In this section, we provide algorithms \texttt{LeastSquares} and \texttt{BatchAvgLeastSquares} for recovering the coefficients in a Bayesian network using linear least squares.
As discussed in \cref{sec:two-phased}, we will recover the coefficients for each variable such that \ref{eq:condition1} is satisfied.
To do so, we estimate the coefficients associated with each individual variable using independent samples.
At each node, \texttt{LeastSquares} computes an estimate by solving the linear least squares problem with respect to a collection of sample observations.
In \cref{sec:batchleastsquares}, we generalize this approach via \texttt{BatchAvgLeastSquares} by allowing any interpolation between ``batch size'' and ``number of batches'' -- \texttt{LeastSquares} is a special case of a single batch.
Since each solution to batch can be computed independently before their results are combined, \texttt{BatchAvgLeastSquares} facilitates parallelism.

\subsection{Vanilla least squares}
\label{sec:vanillaleastsquares}

Consider an arbitrary variable $Y$ with $p$ parents.
Using $m_1$ independent samples, we form matrix $X \in \R^{m_1 \times p}$, where the $r^{th}$ \emph{row} consists of sample values $X_1^{(r)}, \ldots, X_p^{(r)}$, and the \emph{column} vector $B = [Y^{(1)}, \ldots, Y^{(m_1)}]^\top \in \R^{m_1}$.
Then, we define $\wh{A} = (X^\top X)^{-1} X^\top B$ as the solution to the least squares problem $X \wh{A} = B$.
The pseudocode of \texttt{LeastSquares} is given in \cref{alg:least-squares} and \cref{thm:two-phased-leastsquares} states its guarantees.

\begin{algorithm}[htbp]
\caption{\texttt{LeastSquares}: Coefficient recovery algorithm for general Bayesian networks}
\label{alg:least-squares}
\begin{algorithmic}[1]
    \State \textbf{Input}: DAG $G$ and $m_1 \in \cO\Paren{\frac{n d_{avg}}{\eps} \cdot \ln\Paren{\frac{n}{\delta}}}$ samples
    \For{variable $Y$ with $p \geq 1$ parents}
        \State Without loss of generality, by renaming variables, let $X_1, \ldots, X_p$ be the parents of $Y$.
        \State Form matrix $X \in \R^{m_1 \times p}$, where the $r^{th}$ row consists of sample values $X_1^{(r)}, \ldots, X_p^{(r)}$
        \State Form column vector $B = [Y^{(1)}, \ldots, Y^{(m_1)}]^\top \in \R^{m_1}$
        \State Define $\wh{A} = (X^{\top}X)^{-1}X^{\top} B$ as the solution to the least squares problem $X \wh{A} = B$
    \EndFor
\end{algorithmic}
\end{algorithm}

\begin{theorem}[Distribution learning using \texttt{LeastSquares}]
\label{thm:two-phased-leastsquares}
Let $\eps, \delta \in (0,1)$.
Suppose $G$ is a fixed directed acyclic graph on $n$ variables with degree at most $d$.
Given $\cO\Paren{\frac{n d_{avg}}{\eps^2} \log\Paren{\frac{n}{\delta}}}$ samples from an unknown Bayesian network $\cP$ over $G$, if we use \texttt{LeastSquares} for coefficient recovery in \cref{alg:recovery-two-phased}, then with probability at least $1 - \delta$, we recover a Bayesian network $\cQ$ over $G$ such that $\tv(\cP, \cQ) \leq \eps$ in $\cO\Paren{\frac{n^2 d_{avg}^2 d}{\eps^2} \log \Paren{\frac{1}{\delta}}}$ time.\footnote{In particular, this gives a $\kl(\cP,\cQ)\le \epsilon$ guarantee for learning centered multivariate Gaussians using $\widetilde{O}(n^2\epsilon^{-1})$ samples. See e.g.~\cite{ashtiani2020near} for an analogous $\kl(\cQ,\cP)\le \epsilon$ guarantee.}
\end{theorem}

Our analysis begins by proving guarantees for an arbitrary variable.

\begin{lemma}
\label{lem:least-squares-single}
Consider \cref{alg:least-squares}.
Fix an arbitrary variable $Y$ with $p$ parents, parameters $(A, \sigma_y)$, and associated covariance matrix $M$.
With $k \geq \frac{4c_2^2}{(1 - c_1)^4} \cdot \frac{nd_{avg}}{\eps}$ samples, for any constants $0 < c_1 < 1/2$ and $c_2 > 0$, we recover the coefficients $\wh{A}$ such that
\[
\Pr\Paren{ \Abs{\Delta^\top M \Delta} \geq \sigma_y^2 \cdot \frac{p\eps}{nd_{avg}} }
\leq \exp\Paren{ - \frac{k c_1^2}{2} } + 2p \exp\Paren{ -2 k } + p \exp\Paren{-\frac{c_2^2}{2}}
\]
\end{lemma}
\begin{proof}
Since $\Abs{\Delta^\top M \Delta} = \Abs{\Delta^\top LL^\top \Delta} = \Norm{L^\top \Delta}^2$, it suffices to bound $\Norm{L^\top \Delta}$.

Without loss of generality, the parents of $Y$ are $X_1, \ldots, X_p$.
Define $X \in \R^{k \times p}$, $B \in \R^k$, and $\wh{A} \in \R^p$ as in \cref{alg:least-squares}.
Let $\eta = [\eta_y^{(1)}, \ldots, \eta_y^{(k)}] \in \R^k$ be the instantiations of Gaussian $\eta_y$ in the $k$ samples.
By the structural equations, we know that $B = XA + \eta$.
So,
\[
\wt{A} 
= (X^\top X)^{-1} X^\top B
= (X^\top X)^{-1} X^\top (XA + \eta)
= A + (X^\top X)^{-1} X^\top \eta
\]

By \cref{fact:multivariate-transformation}, we can express $X = GL^\top$ where matrix $G \in \R^{k \times p}$ is a random matrix with i.i.d.\ $N(0,1)$ entries.
Since $\Delta = \wh{A}-A$, we see that $\Delta = (L^\top)^{-1}(G^\top G)^{-1} G^\top \eta$.
Rearranging, we have $L^\top \Delta = (G^\top G)^{-1} G^\top \eta$ and so $\| L^\top \Delta \| \leq \| (G^\top G)^{-1} \| \cdot \| G^\top \eta \|$.
Combining \cref{lem:gtg} and \cref{lem:Getanorm}, which bound $\| (G^\top G)^{-1} \|$ and $\| G^\top \eta \|$ respectively, we get
\begin{equation}
\label{eq:generic-bound}
\Pr\Paren{ \| L^\top \Delta \| > \dfrac{2 \sigma_y c_2 \sqrt{ p}}{\Paren{1 - 2 c_1}^2 \sqrt{k}}}
\leq \exp\Paren{ - \frac{k c_1^2}{2} } + 2p \exp\Paren{ -2 k } + p \exp\Paren{-\frac{c_2^2}{2}}    
\end{equation}
for any constants $0 < c_1 < 1/2$ and $c_2 > 0$.
The claim follows by setting $k = \frac{4c_2^2}{(1 - c_1)^4} \cdot \frac{nd_{avg}}{\eps}$.
\end{proof}

We can now establish \ref{eq:condition1} of \cref{prop:decomposition-implication} for \texttt{LeastSquares}.

\begin{lemma}
\label{lem:least-squares}
Consider \cref{alg:least-squares}.
With $m_1 \in \cO\Paren{\frac{n d_{avg}}{\eps} \cdot \ln\Paren{\frac{n}{\delta}}}$ samples, we recover the coefficients $\wh{A}_1, \ldots, \wh{A}_n$ such that
\[
\Pr\Paren{ \forall i \in [n], \Abs{\Delta_i^\top M_i \Delta_i} \geq \sigma_i^2 \cdot \frac{\eps p_i}{n d_{avg}} }
\leq \delta
\]
The total running time is $\cO\Paren{\frac{n^2 d_{avg}^2 d}{\eps} \ln\Paren{\frac{1}{\delta}}}$.
\end{lemma}
\begin{proof}By setting $c_1 = 1/4$, $c_2 = \sqrt{2 \ln\Paren{3n/\delta}}$, and $k = \frac{32n d_{avg}}{\eps} \ln\Paren{\frac{3n}{\delta}} \geq \frac{4c_2^2}{(1 - c_1)^4} \cdot \frac{nd_{avg}}{\eps}$ in \cref{lem:least-squares-single}, we have
\[
\Pr\Paren{ \Abs{\Delta_i^\top M_i \Delta_i} \geq \sigma_i^2 \cdot \frac{p_i \eps}{n d_{avg}} }
\leq \exp\Paren{ - \frac{k c_1^2}{2} } + p \exp\Paren{ -2 k } + p \exp\Paren{-\frac{c_2^2}{2}}
\leq \frac{\delta}{3n} + \frac{\delta}{3n} + \frac{\delta}{3n}
= \frac{\delta}{n}
\]
for any $i \in [n]$.
The claim holds by a union bound over all $n$ variables.

The computational complexity for a variable with $p$ parents is $\cO(p^2 \cdot m_1)$.
Since $\max_{i \in [n]} p_i \leq d$ and $\sum_{i=1}^n p_i = n d_{avg}$, the total runtime is $\cO(m_1 \cdot n \cdot d_{avg} \cdot d)$.
\end{proof}

\cref{thm:two-phased-leastsquares} follows from combining the guarantees of \texttt{LeastSquares} and \texttt{VarianceRecovery} (given in \cref{lem:least-squares} and \cref{thm:recovery-variance} respectively) via \cref{prop:decomposition-implication}.

\begin{proof}[Proof of \cref{thm:two-phased-leastsquares}]
We will show sample and time complexities before giving the proof for the $\tv$ distance.

Let $m_1 \in \cO\Paren{\frac{n d_{avg}}{\eps} \cdot \ln\Paren{\frac{n}{\delta}}}$ and $m_2 \in \cO\Paren{\frac{n d_{avg}}{\eps} \log\Paren{\frac{n}{\delta}}}$.
Then, the total number of samples needed is $m = m_1 + m_2 \in \cO\Paren{\frac{n d_{avg}}{\eps} \log\Paren{\frac{n}{\delta}}}$.
\texttt{LeastSquares} runs in $\cO\Paren{\frac{n^2 d_{avg}^2 d}{\eps} \ln\Paren{\frac{1}{\delta}}}$ time while \texttt{VarianceRecovery} runs in $\cO\Paren{\frac{n^2 d_{avg}^2}{\eps} \log \Paren{\frac{1}{\delta}}}$ time.
Therefore, the overall running time is $\cO\Paren{\frac{n^2 d_{avg}^2 d}{\eps} \log \Paren{\frac{1}{\delta}}}$.

By \cref{lem:least-squares}, \texttt{LeastSquares} recovers coefficients $\wh{A}_1, \ldots, \wh{A}_n$ such that
\[
\Pr\Paren{ \forall i \in [n], \Abs{\Delta_i^\top M_i \Delta_i} \geq \sigma_i^2 \cdot \frac{\eps p_i}{n d_{avg}} }
\leq \delta
\]

By \cref{thm:recovery-variance} and using the recovered coefficients from \texttt{LeastSquares}, \texttt{VarianceRecovery} recovers variance estimates $\wh{\sigma}_i^2$ such that
\[
\Pr\Paren{\forall i \in [n], \Paren{1 - \sqrt{\frac{\eps p_i}{n d_{avg}}}} \cdot \sigma_i^2 \leq \wh{\sigma}_i^2 \leq \Paren{1 + \sqrt{\frac{\eps p_i}{n d_{avg}}}} \cdot \sigma_i^2} \geq 1 - \delta
\]

As our estimated parameters satisfy \ref{eq:condition1} and \ref{eq:condition2}, \cref{prop:decomposition-implication} tells us that $\kl(\cP,\cQ) \leq 3 \eps$.
Thus, $\tv(\cP,\cQ) \leq \sqrt{\kl(\cP,\cQ) / 2} \leq \sqrt{3 \eps/2}$.
The claim follows by setting $\eps' = \sqrt{3 \eps/2}$ throughout.
\end{proof}

\subsection{Interpolating between batch size and number of batches}
\label{sec:batchleastsquares}

We now discuss a generalization of \texttt{LeastSquares}.
In a nutshell, for each variable with $p \geq 1$ parents, \texttt{BatchAvgLeastSquares} solves $b \geq 1$ batches of linear systems made up of $k > p$ samples and then uses the \emph{mean} of the recovered solutions as an estimate for the coefficients.
Note that one can interpolate between different values of $k$ and $b$, as long as $k \geq p$ (so that the batch solutions are correlated to the true parameters).
The pseudocode of \texttt{BatchAvgLeastSquares} is provided in \cref{alg:batch-least-squares} and the guarantees are given in \cref{thm:two-phased-leastsquares}.

\begin{algorithm}[htbp]
\caption{\texttt{BatchAvgLeastSquares}: Coefficient recovery for general Bayesian networks}
\label{alg:batch-least-squares}
\begin{algorithmic}[1]
    \State \textbf{Input}: DAG $G$ and $m_1 = \cO\Paren{\frac{nd_{avg}}{\eps} \Paren{d+ \ln\Paren{\frac{n}{\eps \delta}}}}$ samples \Comment{$k \in \Omega\Paren{d+ \ln\Paren{\frac{n}{\eps\delta}}}$, $b = \frac{m_1}{k}$}
    \For{variable $Y$ with $p \geq 1$ parents}
        \State Without loss of generality, by renaming variables, let $X_1, \ldots, X_p$ be the parents of $Y$.
        \For{$s = 1, \ldots, b$}
            \State Form matrix $X \in \R^{k \times p}$, where the $r^{th}$ row consists of sample values $X_1^{(s,r)}, \ldots, X_p^{(s,r)}$
            \State Form column vector $B = [Y^{(s,1)}, \ldots, Y^{(s,k)}]^\top \in \R^{k}$
            \State Define $\wt{A}^{(s)} = (X^{\top}X)^{-1}X^{\top} B$ as the solution to the least squares problem $X \wt{A}^{(s)} = B$
        \EndFor
        \State Define $\wh{A} = \frac{1}{b} \sum_{s=1}^b \wt{A}^{(s)}$
    \EndFor
\end{algorithmic}
\end{algorithm}

In \cref{sec:experiments}, we also experimented on a variant of \texttt{BatchAvgLeastSquares} dubbed \texttt{BatchMedLeastSquares}, where $\wh{A}$ is defined to be the coordinate-wise \emph{median} of the $\wt{A}^{(s)}$ vectors. However, in the theoretical analysis below, we only analyze \texttt{BatchAvgLeastSquares}.

\begin{theorem}[Distribution learning using \texttt{BatchAvgLeastSquares}]
\label{thm:two-phased-batchleastsquares}
Let $\eps, \delta \in (0,1)$.
Suppose $G$ is a fixed directed acyclic graph on $n$ variables with degree at most $d$.
Given $\cO\Paren{\frac{nd_{avg}}{\eps^2} \Paren{d+ \ln\Paren{\frac{n}{\eps \delta}}}}$ samples from an unknown Bayesian network $\cP$ over $G$, if we use \texttt{BatchAvgLeastSquares} for coefficient recovery in \cref{alg:recovery-two-phased}, then with probability at least $1 - \delta$, we recover a Bayesian network $\cQ$ over $G$ such that $\tv(\cP, \cQ) \leq \eps$ in $\cO\Paren{\frac{n^2 d_{avg}^2 d}{\eps^2} \Paren{d + \ln\Paren{\frac{n}{\eps \delta}}}}$ time. 
\end{theorem}

Our approach for analyzing \texttt{BatchAvgLeastSquares} is the same as our approach for \texttt{LeastSquares}: we prove guarantees for an arbitrary variable and then take union bound over $n$ variables. At a high-level, for each node $Y$, for every fixing of the randomness in generating $X_1, \dots, X_p$, we show that each $\wt{A}^{(s)}$ is a gaussian. Since the $b$ iterations are independent, $\frac1b \sum_s \wt{A}^{(s)}$ is also a gaussian. Its variance is itself a random variable but can be bounded with high probability using concentration inequalities.

\begin{lemma}
\label{lem:BatchAvgLeastSquares-single}
Consider \cref{alg:batch-least-squares}.
Fix any arbitrary variable of interest $Y$ with $p$ parents, parameters $(A, \sigma_y)$, and associated covariance matrix $M$.
With $k > C_k \cdot \Paren{p + \ln\Paren{\frac{n}{\eps\delta}}}$ and $kb = C_{kb} \cdot \Paren{\frac{nd_{avg}}{\eps} \Paren{d+ \ln\Paren{\frac{n}{\eps\delta}}}}$, for some universal constants $C_k$ and $C_{kb}$, we recover coefficients estimates $\wh{A}$ such that
\[
\Pr\Paren{ \Abs{\Delta^\top M \Delta} \leq \sigma_y^2 \cdot \frac{\eps p}{n d_{avg}}} \geq 1 - \delta
\]
\end{lemma}
\begin{proof}
Without loss of generality, the parents of $Y$ are $X_1, \ldots, X_p$.
For $s \in [n]$, define $X^{(s)} \in \R^{k \times p}$, $B^{(s)} \in \R^k$, and $\wt{A}^{(s)} \in \R^p$ as the quantities involved in the $s^{th}$ batch of \cref{alg:batch-least-squares}.
Let $\eta^{(s)} = [\eta_y^{(s,1)}, \ldots, \eta_y^{(s,k)}] \in \R^k$ be the instantiations of Gaussian $\eta_y$ in the $k$ samples for the $s^{th}$ batch.
By the structural equations, we know that $B^{(s)} = X^{(s)} A + \eta^{(s)}$.
So,
\begin{align*}
\wt{A}^{(s)}
&= ((X^{(s)})^\top X^{(s)})^{-1} (X^{(s)})^\top B\\
&= ((X^{(s)})^\top X^{(s)})^{-1} (X^{(s)})^\top (X^{(s)} A + \eta^{(s)})\\
&= A + ((X^{(s)})^\top X^{(s)})^{-1} (X^{(s)})^\top \eta^{(s)}
\end{align*}

By \cref{fact:multivariate-transformation}, we can express $X^{(s)} = G^{(s)} L^\top$ where matrix $G^{(s)} \in \R^{k \times p}$ is a random matrix with i.i.d.\ $N(0,1)$ entries.
So, we see that
\[
L^\top \Delta
= \frac{1}{b} \sum_{s=1}^b ((G^{(s)})^\top G^{(s)})^{-1} (G^{(s)})^\top \eta^{(s)}
\]

For any $i \in [p]$, \cref{fact:linear-combination-of-gaussians} tells us that
\[
\Paren{L^\top \Delta}_i
= \frac{1}{b} \sum_{s=1}^b \Paren{((G^{(s)})^\top G^{(s)})^{-1} (G^{(s)})^\top \eta^{(s)}}_i
\sim N\Paren{0, \frac{\sigma_y^2}{b^2} \sum_{s=1}^b \Norm{\Paren{((G^{(s)})^\top G^{(s)})^{-1} (G^{(s)})^\top}_i}^2 }
\]

We can upper bound each $\Norm{\Paren{((G^{(s)})^\top G^{(s)})^{-1} (G^{(s)})^\top}_i}$ term as follows:
\[
\Norm{\Paren{((G^{(s)})^\top G^{(s)})^{-1} (G^{(s)})^\top}_i}
\leq \Norm{((G^{(s)})^\top G^{(s)})^{-1} (G^{(s)})^\top}
\leq \Norm{((G^{(s)})^\top G^{(s)})^{-1}} \cdot \Norm{G^{(s)}}
\]

When $k \geq 4p$, \cref{lem:gtg} tells us that
$
\Pr\Paren{ \Norm{((G^{(s)})^\top G^{(s)})^{-1}} \geq \frac{4}{k} } \leq \exp\Paren{- \frac{k}{32}}
$.
Meanwhile, \cref{sec:norm-of-rectangular-gaussian-matrix} tells us that
$
\Pr\Paren{\Norm{G^{(s)}} \geq 2(\sqrt{k} + \sqrt{p})} \leq \sqrt{p} \cdot \exp\Paren{-C \cdot k}
$
for some universal constant $C$.
Let $\mathcal{E}$ be the event that $\Norm{\Paren{((G^{(s)})^\top G^{(s)})^{-1} (G^{(s)})^\top}_i} < \frac{8 (\sqrt{k} + \sqrt{p})}{k}$ for any $s \in [b]$.
Applying union bound with the conclusions from \cref{lem:gtg} and \cref{sec:norm-of-rectangular-gaussian-matrix}, we have
\begin{align*}
\Pr\Paren{\overline{\mathcal{E}}}
&= \Pr\Paren{\exists s \in [b], \Norm{\Paren{((G^{(s)})^\top G^{(s)})^{-1} (G^{(s)})^\top}_i} \geq \frac{8 (\sqrt{k} + \sqrt{p})}{k}}\\
&\leq b \cdot \exp\Paren{- \frac{k}{32}} + b \cdot \sqrt{p} \cdot \exp\Paren{-C \cdot k}
\end{align*}

Conditioned on event $\mathcal{E}$, standard Gaussian tail bounds (e.g.\ see \cref{fact:standard-gaussian-tail-bound}) give us
\begin{align*}
\Pr\Paren{\Abs{L^\top \Delta}_i > \sigma_y \cdot \sqrt{\frac{\eps}{n d_{avg}}} \mathrel{\Big|} \mathcal{E}}
&\leq \exp\Paren{-\frac{\sigma_y^2 \cdot \frac{\eps}{n d_{avg}}}{2 \cdot \frac{\sigma_y^2}{b^2} \sum_{s=1}^b \Norm{\Paren{((G^{(s)})^\top G^{(s)})^{-1} (G^{(s)})^\top}_i}^2}}\\
&\leq \exp\Paren{- \frac{\eps \cdot b \cdot k^2}{128 \cdot n \cdot d_{avg} \cdot (\sqrt{k} + \sqrt{p})^2}}\\
&\leq \exp\Paren{- \frac{\eps \cdot b \cdot k}{512 \cdot n \cdot d_{avg}}}
\end{align*}
where the second last inequality is because of event $\mathcal{E}$ and the last inequality is because $(\sqrt{k} + \sqrt{p})^2 \leq (2\sqrt{k})^2 = 4k$, since $k \geq p$.

Thus, applying a union bound over all $p$ entries of $L^\top \Delta$ and accounting for $\Pr(\overline{\mathcal{E}})$, we have
\begin{align*}
&\; \Pr\Paren{\Norm{L^\top \Delta} > \sigma_y \cdot \sqrt{\frac{\eps p}{n d_{avg}}}}\\
\leq &\; \Pr\Paren{\Norm{L^\top \Delta} > \sigma_y \cdot \sqrt{\frac{\eps p}{n d_{avg}}} \mathrel{\Big|} \mathcal{E}} + \Pr\Paren{\mathcal{\overline{E}}}\\
\leq &\; p \cdot \exp\Paren{- \frac{\eps \cdot b \cdot k}{512 \cdot n \cdot d_{avg}}} + b \cdot \exp\Paren{- \frac{k}{32}} + b \cdot \sqrt{p} \cdot \exp\Paren{-C \cdot k}
\end{align*}
for some universal constant $C$.

The claim follows by observing that $\Abs{\Delta^\top M \Delta} = \Abs{\Delta^\top LL^\top \Delta} = \Norm{L^\top \Delta}^2$ and applying the assumptions on $k$ and $b$.
\end{proof}

We can now establish \ref{eq:condition1} of \cref{prop:decomposition-implication} for \texttt{BatchAvgLeastSquares}.

\begin{lemma}[Coefficient recovery guarantees of \texttt{BatchAvgLeastSquares}]
\label{lem:recovery-BatchAvgLeastSquares}
Consider \cref{alg:batch-least-squares}.
With $m_1 \in \cO\Paren{\frac{nd_{avg}}{\eps} \Paren{d + \ln\Paren{\frac{n}{\eps \delta}}}}$ samples, where $k \in \Omega\Paren{d + \ln\Paren{\frac{n}{\eps\delta}}}$ and $b = \frac{m_1}{k}$, we recover coefficient estimates $\wh{A}_1, \ldots, \wh{A}_n$ such that
\[
\Pr\Paren{ \forall i \in [n], \Abs{\Delta_i^\top M_i \Delta_i} \geq \sigma_i^2 \cdot \frac{\eps p_i}{n d_{avg}} } \leq \delta
\]
The total running time is $\cO(m_1 \cdot n \cdot d_{avg} \cdot d)$.
\end{lemma}
\begin{proof}
For each $i \in [n]$, apply \cref{lem:BatchAvgLeastSquares-single} with $\delta' = \delta/n$, then take the union bound over all $n$ variables.

The computational complexity for a variable with $p$ parents is $\cO(b \cdot p^2 \cdot k) = \cO(p^2 \cdot m_1)$.
Since $\max_{i \in [n]} p_i \leq d$ and $\sum_{i=1}^n p_i = n d_{avg}$, the total runtime is $\cO(m_1 \cdot n \cdot d_{avg} \cdot d)$.
\end{proof}

\cref{thm:two-phased-batchleastsquares} follows from combining the guarantees of \texttt{BatchAvgLeastSquares} and \texttt{VarianceRecovery} (given in \cref{lem:recovery-BatchAvgLeastSquares} and \cref{thm:recovery-variance} respectively) via \cref{prop:decomposition-implication}.

\begin{proof}[Proof of \cref{thm:two-phased-batchleastsquares}]
We will show sample and time complexities before giving the proof for the $\tv$ distance.

Let $m_1 \in \cO\Paren{\frac{nd_{avg}}{\eps} \Paren{d + \ln\Paren{\frac{n}{\eps \delta}}}}$ and $m_2 \in \cO\Paren{\frac{n d_{avg}}{\eps} \log\Paren{\frac{n}{\delta}}}$.
Then, the total number of samples needed is $m = m_1 + m_2 \in \cO\Paren{\frac{nd_{avg}}{\eps} \Paren{d+ \ln\Paren{\frac{n}{\eps \delta}}}}$.
\texttt{BatchAvgLeastSquares} runs in $\cO(m_1 n d_{avg} d)$ time while \texttt{VarianceRecovery} runs in $\cO\Paren{\frac{n^2 d_{avg}^2}{\eps} \log \Paren{\frac{1}{\delta}}}$ time.
Therefore, the overall running time is $\cO\Paren{\frac{n^2 d_{avg}^2 d}{\eps} \Paren{d + \ln\Paren{\frac{n}{\eps \delta}}}}$.

By \cref{lem:recovery-BatchAvgLeastSquares}, \texttt{BatchAvgLeastSquares} recovers coefficients $\wh{A}_1, \ldots, \wh{A}_n$ such that
\[
\Pr\Paren{ \forall i \in [n], \Abs{\Delta_i^\top M_i \Delta_i} \geq \sigma_i^2 \cdot \frac{\eps p_i}{n d_{avg}} }
\leq \delta
\]

By \cref{thm:recovery-variance} and using the recovered coefficients from \texttt{BatchAvgLeastSquares}, \texttt{VarianceRecovery} recovers variance estimates $\wh{\sigma}_i^2$ such that
\[
\Pr\Paren{\forall i \in [n], \Paren{1 - \sqrt{\frac{\eps p_i}{n d_{avg}}}} \cdot \sigma_i^2 \leq \wh{\sigma}_i^2 \leq \Paren{1 + \sqrt{\frac{\eps p_i}{n d_{avg}}}} \cdot \sigma_i^2} \geq 1 - \delta
\]

As our estimated parameters satisfy \ref{eq:condition1} and \ref{eq:condition2}, \cref{prop:decomposition-implication} tells us that $\kl(\cP,\cQ) \leq 3 \eps$.
Thus, $\tv(\cP,\cQ) \leq \sqrt{\kl(\cP,\cQ) / 2} \leq \sqrt{3 \eps/2}$.
The claim follows by setting $\eps' = \sqrt{3 \eps/2}$ throughout.
\end{proof}

\section{Coefficient recovery algorithm based on Cauchy random variables}
\label{sec:cauchyest}

In this section, we provide novel algorithms \texttt{CauchyEst} and \texttt{CauchyEstTree} for recovering the coefficients in polytree Bayesian networks.
We will show that \texttt{CauchyEstTree} has near-optimal sample complexity, and later in \cref{sec:experiments}, we will see that both these algorithms outperform \texttt{LeastSquares} and \texttt{BatchAvgLeastSquares} on randomly generated Bayesian networks.
Of technical interest, our analysis involves Cauchy random variables, which are somewhat of a rarity in statistical learning.
As in \texttt{LeastSquares} and \texttt{BatchAvgLeastSquares}, \texttt{CauchyEst} and \texttt{CauchyEstTree} use independent samples to recover the coefficients associated to each individual variable in an independent fashion.

Consider an arbitrary variable $Y$ with $p$ parents.
The intuition is as follows: if $\eta_y = 0$, then one can form a linear system of equations using $p$ samples to solve for the coefficients $a_{y \leftarrow i}$ \emph{exactly} for each $i \in \pi(y)$.
Unfortunately, $\eta_y$ is non-zero in general.
Instead of exactly recovering $A$, we partition the $m_1$ independent samples into $k = \lfloor m_1/p \rfloor$ batches involving $p$ samples and form intermediate estimates $\wt{A}^{(1)}, \ldots, \wt{A}^{(k)}$ by solving a system of linear equation for each batch (see \cref{alg:batch}).
Then, we ``combine'' these intermediate estimates to obtain our estimate $\wh{A}$.

\begin{algorithm}[htbp]
\caption{Batch coefficient recovery algorithm for variable with $p$ parents}
\label{alg:batch}
\begin{algorithmic}[1]
    \State \textbf{Input}: DAG $G$, a variable $Y$ with $p$ parents, and $p$ samples
    \State Without loss of generality, by renaming variables, let $X_1, \ldots, X_p$ be the parents of $Y$.
    \State Form matrix $X \in \R^{p \times p}$ where the $r^{th}$ \emph{row} equals $[X_1^{(r)}, \ldots, X_p^{(r)}]$, corresponding to $Y^{(r)}$.
    \State Define $\wt{A} = \Brac{\wh{a}_{y \leftarrow 1}, \ldots, \wh{a}_{y \leftarrow p}}^\top$ as \emph{any} solution to $X \wt{A} = \Brac{Y^{(1)}, \ldots, Y^{(p)}}^\top$.
    \State \Return $\wt{A}$
\end{algorithmic}
\end{algorithm}

Consider an arbitrary copy of recovered coefficients $\wt{A}$.
Let $\Delta = \Brac{\Delta_1, \ldots, \Delta_p}^\top = \wt{A} - A$ be a vector measuring the gap between these recovered coefficients and the ground truth, where $\Delta_i = \wt{a}_{y \leftarrow i} - a_{y \leftarrow i}$.
\cref{lem:term-wise-cauchy} gives a condition where a vector is term-wise Cauchy.
Using this, \cref{lem:L-Delta-Cauchy} shows that each entry of the vector $L^\top \Delta$ is distributed according to $\sigma_y \cdot \cau(0,1)$, although the entries may be correlated with each other in general.

\begin{lemma}
\label{lem:term-wise-cauchy}
Consider the matrix equation $AB = E$ where $A \in \R^{n \times n}$, $B \in \R^{n \times 1}$, and $E \in R^{n \times 1}$ such that entries of $A$ and $E$ are independent Gaussians, elements in each \emph{column} of $A$ have the same variance, and all entries in $E$ have the same variance.
That is, $A_{\cdot, j} \sim N(0, \sigma_i^2)$ and $E_i \sim N(0, \sigma_{n+1}^2)$.
Then, for all $i \in [n]$, we have that $B_i \sim \frac{\sigma_{n+1}}{\sigma_i} \cdot \cau(0,1)$.
\end{lemma}
\begin{proof}
As the event that $A$ is singular has measure zero, we can write $B = A^{-1} E$.
By Cramer's rule,
\[
A^{-1} = \frac{1}{\det(A)} \cdot {\rm adj}(A) = \frac{1}{\det(A)} \cdot C^\top
\]
where $\det(A)$ is the determinant of $A$, ${\rm adj}(A)$ is the adjugate/adjoint matrix of $A$, and $C$ is the cofactor matrix of $A$.
Recall that the $\det(A)$ can defined with respect to elements in $C$: For any \emph{column} $i \in [n]$,
\[
\det(A) = A_{1,i} \cdot C_{1,i} + A_{2,i} \cdot C_{2,i} + \ldots + A_{n,i} \cdot C_{n,i}
\]
So, $\det(A) \sim N\Paren{0, \sigma_i^2 \Paren{C_{1,i} + \ldots + C_{n,i}}}$.
Thus, for any $i \in [n]$,
\[
B_i
= \Paren{\frac{1}{\det(A)} C^\top E}_i
\sim \frac{N\Paren{0, \sigma_{n+1}^2 \Paren{C_{1,i} + \ldots + C_{n,i}}}}{N\Paren{0, \sigma_{i}^2 \Paren{C_{1,i} + \ldots + C_{n,i}}}}
= \frac{\sigma_{n+1}}{\sigma_i} \cdot \cau(0,1)
\]
\end{proof}

\begin{lemma}
\label{lem:L-Delta-Cauchy}
Consider a batch estimate $\wt{A}$ from \cref{alg:batch}.
Then, $L^\top \Delta$ is entry-wise distributed as $\sigma_y \cdot \cau(0,1)$, where $\Delta = \wt{A} - A$.
Note that the entries of $L^\top \Delta$ may be correlated in general.
\end{lemma}
\begin{proof}
Observe that each row of $X$ is an independent sample drawn from a multivariate Gaussian $N(0, M)$.
By denoting $\eta = \Brac{\eta_{y}^{(1)}, \ldots, \eta_{y}^{(p)}}^\top$ as the $p$ samples of $\eta_y$, we can write $X \wt{A} = X A + \eta$ and thus $X \Delta = \eta$ by rearranging terms.
By \cref{fact:multivariate-transformation}, we can express $X = GL^\top$ where matrix $G \in \R^{p \times p}$ is a random matrix with i.i.d.\ $N(0,1)$ entries.
By substituting $X = GL^\top$ into $X \Delta = \eta$, we have $L^\top \Delta = G^{-1} \eta$.\footnote{Note that event that $G$ is singular has measure 0.}

By applying \cref{lem:term-wise-cauchy} with the following parameters: $A = G, B = L^\top \Delta, E = \eta$, we conclude that each entry of $L^\top \Delta$ is distributed as $\sigma_y \cdot \cau(0,1)$.
However, note that \emph{these entries are generally correlated}.
\end{proof}

If we have direct access to the matrix $L$, then one can do the following (see \cref{alg:cauchyest}): for each coordinate $i \in [p]$, take \emph{medians}\footnote{The typical strategy of averaging independent estimates does not work here as the variance of a Cauchy variable is unbounded.
} of $\Paren{L^\top \Brac{\wt{a}_{y \leftarrow 1}, \ldots, \wt{a}_{y \leftarrow n}}^\top}_i$ to form $\texttt{MED}_i$ and then estimate $\Brac{\wh{a}_{y \leftarrow 1}, \ldots, \wh{a}_{y \leftarrow 1}} = (L^\top)^{-1} [\texttt{MED}_1, \ldots, \texttt{MED}_n]^\top$.
By the convergence of Cauchy random variables to their median, one can show that each $\wh{a}_{y \leftarrow i}$ converges to the true coefficient $a_{y \leftarrow i}$ as before.
Unfortunately, we do not have $L$ and can only hope to estimate it with some matrix $\wh{L}$ using the \emph{empirical} covariance matrix $\wh{M}$.

\begin{algorithm}[htbp]
\caption{\texttt{CauchyEst}: Coefficient recovery algorithm for general Bayesian networks}
\label{alg:cauchyest}
\begin{algorithmic}[1]
    \State \textbf{Input}: DAG $G$ and $m$ samples
    \For{variable $Y$ with $p \geq 1$ parents} \Comment{By \cref{assumption:d-parents}, $p \leq d$}
        \State Without loss of generality, by renaming variables, let $X_1, \ldots, X_p$ be the parents of $Y$.
        \State Let $\wh{M}$ be the empirical covariance matrix with respect to $X_1, \ldots, X_p$.
        \State Compute the Cholesky decomposition $\wh{M} = \wh{L} \wh{L}^\top$ of $\wh{M}$.
        \For{$s = 1, \ldots, \lfloor m/p \rfloor$}
        	\State Using $p$ samples and \cref{alg:batch}, compute a batch estimate $\wt{A}^{(s)}$.
        \EndFor
        \State For each $i \in [n]$, define $\texttt{MED}_i = \med \{ (\wh{L}^\top \wt{A}^{(1)})_i, \ldots, (\wh{L}^\top \wt{A}^{(\lfloor m/p \rfloor)})_i\}$.
        \State \Return
        $
        \wh{A}
        = \Brac{\wh{a}_{y \leftarrow 1}, \ldots, \wh{a}_{y \leftarrow n}}^\top
        = ( (\wh{L}^\top)^{-1} \Brac{\texttt{MED}_1, \ldots, \texttt{MED}_n}^\top)^\top
        $.
    \EndFor
\end{algorithmic}
\end{algorithm}

\subsection{Special case of polytree Bayesian networks}

If the Bayesian network is a polytree, then $L$ is diagonal.
In this case, we specialize \texttt{CauchyEst} to \texttt{CauchyEstTree} and are able to give theoretical guarantees.
We begin with simple corollary which tells us that the $i^{th}$ entry of $\Delta$ is distributed according to $\sigma_y / \sigma_i \cdot \cau(0,1)$.

\begin{corollary}
\label{cor:L-Delta-Cauchy-tree}
Consider a batch estimate $\wt{A}$ from \cref{alg:batch}.
If the Bayesian network is a polytree, then $\Delta_i = (\wt{A} - A)_i \sim \frac{\sigma_y}{\sigma_i} \cdot \cau(0,1)$.
\end{corollary}
\begin{proof}
Observe that each row of $X$ is an independent sample drawn from a multivariate Gaussian $N(0, M)$.
By denoting $\eta = \Brac{\eta_{y}^{(1)}, \ldots, \eta_{y}^{(p)}}^\top$ as the $p$ samples of $\eta_y$, we can write $X \wt{A} = X A + \eta$ and thus $X \Delta = \eta$ by rearranging terms.
Since the parents of any variable in a polytree are not correlated, each element in the $i^{th}$ column of $X$ is a $N(0, \sigma_i^2)$ Gaussian random variable.

By applying \cref{lem:term-wise-cauchy} with the following parameters: $A = X, B = \Delta E = \eta$, we conclude that $\Delta_i = (\wt{A} - A)_i \sim \frac{\sigma_y}{\sigma_i} \cdot \cau(0,1)$.
\end{proof}

For each $i \in \pi(y)$, we combine the $k$ independently copies of $\wt{a}_{y \leftarrow i}^{(1)}, \ldots, \wt{a}_{y \leftarrow i}^{(k)}$ using the \emph{median}.
For arbitrary sample $s \in [k]$ and parent index $i \in \pi(y)$, observe that $\Delta_i^{(s)} = \wt{a}_{y \leftarrow i}^{(s)} - a_{y \leftarrow i}$.
Since $a_{y \leftarrow i}$ is just an unknown \emph{constant},
\[
\wh{a}_{y \leftarrow i}
= \med_{s \in [k]} \Brace{\wt{a}_{y \leftarrow i}^{(s)}}
= \med_{s \in [k]} \Brace{\Delta_i^{(s)}} + a_{y \leftarrow i}
\]
Since each $\Delta_i^{(s)}$ term is i.i.d.\ distributed as $\sigma_y \cdot \cau(0,1)$, the term $\med_{s \in [k]} \Brace{\Delta_i^{(s)}}$ converges to 0 with sufficiently large $k$, and thus $\wh{a}_{y \leftarrow i}$ converges to the true coefficient $a_{y \leftarrow i}$.

The goal of this section is to prove \cref{thm:two-phased-cauchyesttree} given \texttt{CauchyEstTree}, whose pseudocode we provide in \cref{alg:cauchyesttree}.

\begin{algorithm}[htbp]
\caption{\texttt{CauchyEstTree}: Coefficient recovery algorithm for polytree Bayesian networks}
\label{alg:cauchyesttree}
\begin{algorithmic}[1]
    \State \textbf{Input}: A polytree $G$ and $m_1 \in \cO\Paren{\frac{n d_{avg} d}{\eps} \log\Paren{\frac{n}{\delta}}}$ samples
    \For{variable $Y$ with $p \geq 1$ parents} \Comment{By \cref{assumption:d-parents}, $p \leq d$}
        \State Without loss of generality, by renaming variables, let $X_1, \ldots, X_p$ be the parents of $Y$.
        \For{$s = 1, \ldots, \lfloor m_1/p \rfloor$}
            \State Using $p$ samples and \cref{alg:batch}, compute a batch estimate $\wt{A}^{(s)}$.
        \EndFor
        \State For each $i \in \pi(y)$, define estimate $\wh{a}_{y \leftarrow i} = \med \Brace{\wt{a}_{y \leftarrow i}^{(1)}, \ldots, \wt{a}_{y \leftarrow i}^{(\lfloor m_1/p \rfloor)}}$.
        \State \Return $\wh{A} = [\wh{a}_{y \leftarrow 1}, \ldots, \wh{a}_{y \leftarrow p}]^\top$
    \EndFor
\end{algorithmic}
\end{algorithm}

\begin{theorem}[Distribution learning using \texttt{CauchyEstTree}]
\label{thm:two-phased-cauchyesttree}
Let $\eps, \delta \in (0,1)$.
Suppose $G$ is a fixed directed acyclic graph on $n$ variables with degree at most $d$.
Given $\cO\Paren{\frac{nd_{avg}d}{\eps} \log\Paren{\frac{n}{\eps \delta}}}$ samples from an unknown Bayesian network $\cP$ over $G$, if we use \texttt{CauchyEstTree} for coefficient recovery in \cref{alg:recovery-two-phased}, then with probability at least $1 - \delta$, we recover a Bayesian network $\cQ$ over $G$ such that $\tv(\cP, \cQ) \leq \eps$ in $\cO\Paren{\frac{n^2 d_{avg}^2 d^{\omega-1}}{\eps} \log\Paren{\frac{n}{\delta}}}$ time.
\end{theorem}

Note that for polytrees, $d_{avg}$ is just a constant.
As before, we will first prove guarantees for an arbitrary variable and then take union bound over $n$ variables.

\begin{lemma}
\label{lem:cauchyesttree-single}
Consider \cref{alg:cauchyesttree}.
Fix an arbitrary variable of interest $Y$ with $p$ parents, parameters $(A, \sigma_y)$, and associated covariance matrix $M$.
With $k = \frac{8 n d_{avg}}{\eps} \log\Paren{\frac{2}{\delta}}$ samples, we recover coefficient estimates $\wh{A}$ such that
\[
\Pr\Paren{ \Abs{\Delta^\top M \Delta} \leq \sigma_y^2 \cdot \frac{\eps p}{n d_{avg}} } \geq 1 - \delta
\]
\end{lemma}
\begin{proof}
Since $M = LL^\top$, it suffices to bound $\| L^\top \Delta \|$.
\cref{lem:L-Delta-Cauchy} tells us that each entry of the vector $L^\top \Delta$ is the median of $k$ copies of $\cau(0,1)$ random variables multiplied by $\sigma_y$.
Setting $k = \frac{8 n d_{avg}}{\eps} \log\Paren{\frac{2}{\delta}}$ and $0 < \tau = \sqrt{\eps / (n d_{avg})} < 1$ in \cref{lem:cauchy-median-convergence}, we see that
\[
\Pr\Paren{\text{median of $k$ i.i.d.\ $\cau(0,1)$ random variables} \not\in \Brac{-\sqrt{\frac{\eps}{n d_{avg}}}, \sqrt{\frac{\eps}{n d_{avg}}}}}
\leq \delta
\]
That is, each entry of $L^\top \Delta$ has absolute value at most $\sigma_y \cdot \sqrt{\frac{\eps}{n d_{avg}}}$.
By summing across all $p$ entries of $L^\top \Delta$, we see that
\[
| \Delta^\top M \Delta |
= | \Delta^\top LL^\top \Delta |
= \| L^\top \Delta \|^2
\leq p \cdot \sigma_y^2 \cdot \frac{\eps}{n d_{avg}}
= \sigma_y^2 \cdot \frac{\eps p}{n d_{avg}}
\]
\end{proof}

We can now establish \ref{eq:condition1} of \cref{prop:decomposition-implication} for \texttt{CauchyEstTree}.

\begin{lemma}
\label{lem:cauchyesttree}
Consider \cref{alg:cauchyesttree}.
Suppose the Bayesian network is a polytree.
With $m_1 \in \cO\Paren{\frac{n d_{avg} d}{\eps} \log\Paren{\frac{n}{\delta}}}$ samples, we recover coefficient estimates $\wh{A}_1, \ldots, \wh{A}_n$ such that
\[
\Pr\Paren{ \forall i \in [n], \Abs{\Delta_i^\top M_i \Delta_i} \geq \sigma_i^2 \cdot \frac{\eps p_i}{n d_{avg}} } \leq \delta
\]
The total running time is $\cO\Paren{\frac{n^2 d_{avg}^2 d^{\omega-1}}{\eps} \log\Paren{\frac{n}{\delta}}}$ where $\omega$ is the matrix multiplication exponent.
\end{lemma}
\begin{proof}
For each $i \in [n]$, apply \cref{lem:cauchyesttree-single} with $\delta' = \delta/n$ and $m_1 = \frac{8 n d_{avg}}{\eps} \log\Paren{\frac{2n}{\delta}}$, then take the union bound over all $n$ variables.

The runtime of \cref{alg:batch} is the time to find the inverse of a $p \times p$ matrix, which is $\cO(p^{\omega})$ for some $2< \omega <3$.
Therefore, the computational complexity for a variable with $p$ parents is $\cO(p^{\omega-1} \cdot m_1)$.
Since $\max_{i \in [n]} p_i \leq d$ and $\sum_{i=1}^n p_i = n d_{avg}$, the total runtime is $\cO(m_1 \cdot n \cdot d_{avg} \cdot d^{\omega-2})$.
\end{proof}

We are now ready to prove \cref{thm:two-phased-cauchyesttree}.

\cref{thm:two-phased-cauchyesttree} follows from combining the guarantees of \texttt{CauchyEstTree} and \texttt{VarianceRecovery} (given in \cref{lem:cauchyesttree} and \cref{thm:recovery-variance} respectively) via \cref{prop:decomposition-implication}.

\begin{proof}[Proof of \cref{thm:two-phased-cauchyesttree}]
We will show sample and time complexities before giving the proof for the $\tv$ distance.

Let $m_1 \in \cO\Paren{\frac{n d_{avg} d}{\eps} \log\Paren{\frac{n}{\delta}}}$ and $m_2 \in \cO\Paren{\frac{n d_{avg}}{\eps} \log\Paren{\frac{n}{\delta}}}$.
Then, the total number of samples needed is $m = m_1 + m_2 \in \cO\Paren{\frac{nd_{avg}d}{\eps} \log\Paren{\frac{n}{\eps \delta}}}$.
\texttt{CauchyEstTree} runs in $\cO\Paren{\frac{n^2 d_{avg}^2 d^{\omega-1}}{\eps} \log\Paren{\frac{n}{\delta}}}$ time while \texttt{VarianceRecovery} runs in $\cO\Paren{\frac{n^2 d_{avg}^2}{\eps} \log \Paren{\frac{1}{\delta}}}$ time, where $\omega$ is the matrix multiplication exponent.
Therefore, the overall running time is $\cO\Paren{\frac{n^2 d_{avg}^2 d^{\omega-1}}{\eps} \log\Paren{\frac{n}{\delta}}}$.

By \cref{lem:cauchyesttree}, \texttt{CauchyEstTree} recovers coefficients $\wh{A}_1, \ldots, \wh{A}_n$ such that
\[
\Pr\Paren{ \forall i \in [n], \Abs{\Delta_i^\top M_i \Delta_i} \geq \sigma_i^2 \cdot \frac{\eps p_i}{n d_{avg}} }
\leq \delta
\]

By \cref{thm:recovery-variance} and using the recovered coefficients from \texttt{CauchyEstTree}, \texttt{VarianceRecovery} recovers variance estimates $\wh{\sigma}_i^2$ such that
\[
\Pr\Paren{\forall i \in [n], \Paren{1 - \sqrt{\frac{\eps p_i}{n d_{avg}}}} \cdot \sigma_i^2 \leq \wh{\sigma}_i^2 \leq \Paren{1 + \sqrt{\frac{\eps p_i}{n d_{avg}}}} \cdot \sigma_i^2} \geq 1 - \delta
\]

As our estimated parameters satisfy \ref{eq:condition1} and \ref{eq:condition2}, \cref{prop:decomposition-implication} tells us that $\kl(\cP,\cQ) \leq 3 \eps$.
Thus, $\tv(\cP,\cQ) \leq \sqrt{\kl(\cP,\cQ) / 2} \leq \sqrt{3 \eps/2}$.
The claim follows by setting $\eps' = \sqrt{3 \eps/2}$ throughout.
\end{proof}

\section{Hardness for learning Gaussian Bayesian networks}
\label{sec:hardness}

In this section, we present our hardness results.
We first show a tight lower-bound for the simpler case of learning Gaussian product distributions in total variational distance (\cref{thm:tv-lowerbound}).
Next, we show a tight lower-bound for learning Gaussian Bayes nets with respect to total variation distance. (\cref{thm:kl-lowerbound}).
In both cases, our hardness applies to the problems of learning the covariance matrix of a centered multivariate Gaussian, which is equivalent to recovering the coefficients and noises of the underlying linear structural equation.

We will need the following fact about the variation distance between multivariate Gaussians and a Frobenius norm $\Norm{\cdot}_F$ between the covariance matrices.

\begin{fact}[\cite{devroye2018total}]
\label{fact:tvgauss}
There exists two universal constants $\frac{1}{100}\le c_1\le c_2\le \frac{3}{2}$, such that for any two covariance matrices $\Sigma_1$ and $\Sigma_2$, 
\[
c_1\le \frac{\tv(\gauss(0,\Sigma_1),\gauss(0,\Sigma_2))}{\Norm{\Sigma^{-1}_1\Sigma_2-I}_F} \le c_2.
\]
\end{fact}

\begin{theorem}
\label{thm:tv-lowerbound}
Given samples from a $n$-fold Gaussian product distribution $P$, learning a $\widehat{P}$ such that in $\tv(P,\widehat{P})=O(\epsilon)$ with success probability 2/3 needs $\Omega(n\epsilon^{-2})$ samples in general.
\end{theorem}
\begin{proof}
Let $\mathcal{C}\subseteq \{0,1\}^n$ be a set with the following properties:
1) $|\mathcal{C}|=2^{\Theta(n)}$ and
2) every $i\neq j\in \mathcal{C}$ have a Hamming distance $\Theta(n)$.
Existence of such a set is well-known.
We create a class of Gaussian product distributions $\mathcal{P}_{\mathcal{C}}$ based on $\mathcal{C}$ as follows.
For each $s\in \mathcal{C}$, we define a distribution $P_{s}\in \mathcal{P}_{\mathcal{C}}$ such that if the $i$-th bit of $s$ is 0, we use the distribution $\gauss(0,1)$ for the $i$-th component of $P_s$; else if the $i$-th bit is 1, we use the distribution $\gauss(0,1+\frac{\epsilon}{\sqrt{n}})$.
Then for any $P_s\neq P_t$, $\kl(P_s,P_t)=O(\epsilon^2)$.
Fano's inequality tells us that guessing a random distribution from $\mathcal{P}_{\mathcal{C}}$ correctly with 2/3 probability needs $\Omega(n\epsilon^{-2})$ samples.

\cref{fact:tvgauss} tells us that for any $P_s\neq P_t \in \mathcal{P}_{\mathcal{C}}$, $\tv(P_s,P_t)\ge c_3\epsilon$ for some constant $c_3$. 
Consider any algorithm for learning a random distribution $P=\gauss(0,\Sigma)$ from $\mathcal{P}_{\mathcal{C}}$ in $\tv$-distance at most $c_4\epsilon$ for a sufficiently small constant $c_4$.
Let the learnt distibution be $\wh{P}=\gauss(0,\wh{\Sigma})$.
Due to triangle inequality,~\cref{fact:tvgauss}, and an appropriate choice of $c_4$, $P$ must be the unique distribution from $\mathcal{P}_{\mathcal{C}}$ satisfying $||\wh{\Sigma}^{-1}\Sigma-I||_F\le c_5\epsilon$ for an appropriate choice of $c_5$.
We can find this unique distribution by computing  $||\wh{\Sigma}^{-1}\Sigma'-I||_F$ for every covariance matrix $\Sigma'$ from $\mathcal{P}_{\mathcal{C}}$ and guess the random distribution correctly.
Hence, the lower-bound follows.
\end{proof}

Now, we present the lower-bound for learning general Bayes nets. 

\begin{theorem}
\label{thm:kl-lowerbound}
For any $0<\epsilon<1$ and $n,d$ such that $d \le n/2$, there exists a DAG $G$ over $[n]$ of in-degree $d$ such that learning a Gaussian Bayes net $\wh{P}$ on $G$ such that $\tv(P,\widehat{P})\le \epsilon$ with success probability 2/3 needs $\Omega(nd\epsilon^{-2})$ samples in general.
\end{theorem}
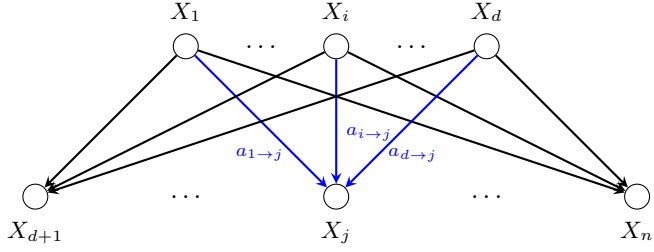
\begin{figure}[t]
\centering
\begin{tikzpicture}
\node[draw, circle, minimum size=5pt] at (-2,0) (top1) {};
\node[circle, minimum size=5pt] at (-1,0) (top2) {$\ldots$};
\node[draw, circle, minimum size=5pt] at (0,0) (top3) {};
\node[circle, minimum size=5pt] at (1,0) (top4) {$\ldots$};
\node[draw, circle, minimum size=5pt] at (2,0) (top5) {};
\node[draw, circle, minimum size=5pt] at (-4,-2) (bot1) {};
\node[circle, minimum size=5pt] at (-2,-2) (bot2) {$\ldots$};
\node[draw, circle, minimum size=5pt] at (0,-2) (bot3) {};
\node[circle, minimum size=5pt] at (2,-2) (bot4) {$\ldots$};
\node[draw, circle, minimum size=5pt] at (4,-2) (bot5) {};

\draw[thick, ->, -stealth] (top1) -- (bot1);
\draw[thick, ->, -stealth, blue] (top1) -- node[left, pos=0.75]{\scriptsize $a_{1 \rightarrow j}$} (bot3);
\draw[thick, ->, -stealth] (top1) -- (bot5);
\draw[thick, ->, -stealth] (top3) -- (bot1);
\draw[thick, ->, -stealth, blue] (top3) -- node[right, pos=0.6]{\scriptsize $a_{i \rightarrow j}$} (bot3);
\draw[thick, ->, -stealth] (top3) -- (bot5);
\draw[thick, ->, -stealth] (top5) -- (bot1);
\draw[thick, ->, -stealth, blue] (top5) -- node[right, pos=0.75]{\scriptsize $a_{d \rightarrow j}$} (bot3);
\draw[thick, ->, -stealth] (top5) -- (bot5);

\node[above=1pt of top1] {\small $X_1$};
\node[above=1pt of top3] {\small $X_i$};
\node[above=1pt of top5] {\small $X_d$};
\node[below=1pt of bot1] {\small $X_{d+1}$};
\node[below=1pt of bot3] {\small $X_j$};
\node[below=1pt of bot5] {\small $X_n$};

\end{tikzpicture}
\caption{Bipartite DAG on $n$ vertices with maximum degree $d$.
For $i \in \{1, \ldots, d\}$, $X_i = \eta_i$ where $\eta_i \sim N(0,1)$.
For $j \in \{d+1, \ldots, n\}$, $X_j = \eta_j + \sum_{i=1}^d a_{i \rightarrow j} X_i$  where $\eta_j \sim N(0,1)$.
Furthermore, each $X_j$ is associated with a $d$-bit string and each coefficients $a_{1 \rightarrow j}, \ldots, a_{d \rightarrow j}$ is either $\frac{1}{\sqrt{d(n-d)}}$ or $\frac{1+\epsilon}{\sqrt{d(n-d)}}$, depending on the $i^{th}$ bit in the associated $d$-bit string.}
\label{fig:lbdag}
\end{figure}
Let $\mathcal{C}\subseteq \{0,1\}^d$ be a set with the following properties:
1) $|\mathcal{C}|=2^{\Theta(d)}$ and
2) every $i\neq j \in \mathcal{C}$ have a Hamming distance $\Theta(d)$.
Existence of such a set is well-known.
We define a class of distributions $\mathcal{P}_{\mathcal{C}}$ based on $\mathcal{C}$ and the graph $G$ shown in \cref{fig:lbdag} as follows.
Each vertex of each distribution in $\mathcal{P}_{\mathcal{C}}$ has a $\gauss(0,1)$ noise, and hence no learning is required for the noises.
Each coefficient $a_{i\rightarrow j}$ takes one of two values $\left\{\frac{1}{\sqrt{d(n-d)}},\frac{1+\epsilon}{\sqrt{d(n-d)}}\right\}$ corresponding to bits $\{0,1\}$ respectively.
For each $s\in \mathcal{C}$, we define $A_s$ to be the vector of coefficients corresponding to the bit-pattern of $s$ as above.
We have $2^{\Theta(d)}$ possible bit-patterns, which we use to define each conditional probability $(X_i\mid X_1,X_2,\dots,X_d)$.
Then, we have a class $\mathcal{Q}_{\mathcal{P}}$ of $|\mathcal{C}|^{(n-d)}$ distributions for the overall Bayes net.
We prune some of the distributions to get the desired subclass $\mathcal{P}_{\mathcal{C}}\subseteq \mathcal{Q}_{\mathcal{C}}$, such that $\mathcal{P}_{\mathcal{C}}$ is the largest-sized subset with any pair of distributions in the subset differing in at least $(n-d)/2$ bit-patterns (out of $(n-d)$ many) for the $(X_i\mid X_1,X_2,\dots,X_d)$'s. 

\begin{claim}
$|\cP_{\cC}|\ge 2^{\Theta(d(n-d))}$.
\end{claim}

\begin{proof}
Let $\ell=|\cC|=2^{\Theta(d)}$ and $m=n-d$.
Then, there are $N=\ell^m$ possible coefficient-vectors/distributions in $\cQ_{\cC}$.
We create an undirectred graph $G$ consisting of a vertex for each coefficient, and edges between any pair of vertices differing in at least $m/2$ bit-patterns.
Note that $G$ is $r$-regular for $r={m\choose 0.5m}(\ell-1)^{0.5m}+\dots+{m\choose 0.5m+j}(\ell-1)^{0.5m+j}+\dots+(\ell-1)^m$.
Turan's theorem says that there is a clique of size $\alpha=(1-\frac{r}{N})^{-1}$. We define $\cP_{\cC}$ to be the vertices of this clique.

To show that $\alpha$ is as large as desired, it suffices to show $N-r\le N/2^{\Theta(dm)}$.
The result follows by noting $N-r=1+{m\choose 1}(\ell-1)+\dots+{m\choose j}(\ell-1)^{j}+\dots+{m\choose 0.5m}(\ell-1)^{0.5m}\le m\cdot(4(\ell-1))^{0.5m}$.
\end{proof}

We also get for any two distributions $P_a,P_b$ in this model with the coefficients $a,b$ respectively, $\kl(P_a,P_b)=\frac{1}{2}||a-b||_2^2$ from \eqref{eq:decomposition}.
Hence, any pair of $\mathcal{P}_{\mathcal{C}}$ has $\kl=\Theta(\epsilon^2)$.
Then, Fano's inequality tells us that given a random distribution $P\in \mathcal{P}_{\mathcal{C}}$, it needs at least $\Omega(d(n-d)\epsilon^{-2})$ samples to guess the correct one with 2/3 probability.
We need the following fact about the $\tv$-distance among the members of $\cP_{\cC}$.

\begin{claim}
Let $P_a, P_b \in \cP_{\cC}$ be two distinct distributions (i.e.\ $P_a \neq P_b$) with coefficient-vectors $a,b$ respectively.
Then, $\tv(P_a,P_b) \in \Theta(\epsilon)$.
\end{claim}
\begin{proof}
By Pinsker's inequality, we have $\tv(P_a,P_b) \in \cO(\epsilon)$.
Here we show $\tv(P_a,P_b) \in \Omega(\epsilon)$.
Let $n-d=m$.
By construction, $a$ and $b$ differ in $m'\ge m/2$ conditional probabilities.
Let $a'\subseteq a,b'\subseteq b$ be the coefficient-vectors on the coordinates where they differ. Let $P_{a'},\Sigma_{a'}$ and $P_{b'},\Sigma_{b'}$ be the corresponding marginal distributions on $(m'+d)$ variables, and their covariance matrices.
In the following, we show $||\Sigma_{a'}^{-1}\Sigma_{b'}-I||_F=\Omega(\epsilon)$. This proves the claim from Fact~\ref{fact:tvgauss}.

Let
$
M_{a'}=
\begin{bmatrix}
0_{m'\times m'} & 0_{m'\times d}\\
A_{d\times m'} & 0_{d\times d}
\end{bmatrix}
$
be the adjacency matrix for $P_a$, where the sources appear last in the rows and columns and in the matrix $A$, each $A_{ij}\in \{\frac{1}{\sqrt{md}},\frac{1+\epsilon}{\sqrt{md}}\}$ denote the coefficient from source $i$ to sink $j$.
Similarly, we define $M_{b'}$ using a coefficient matrix $B_{d\times m'}$. Let $\{A_i: 1\le i\le m'\}$ and $\{B_i: 1\le i\le m'\}$ denote the columns of $A$ and $B$.
Then for every $i$, $A_i$ and $B_i$ differ in at least $\Theta(d)$ coordinates by construction.
 
By definition
$
\Sigma_{b'}=
\begin{bmatrix}
\bullet & B^T\\
B & I_{d\times d}
\end{bmatrix}
$
and
$
\Sigma_{a'}^{-1}=
\begin{bmatrix}
I_{m'\times m'} & -A^T\\
-A & \bullet
\end{bmatrix}
$, where $\bullet$ corresponds to certain matrices not relevant to our discussion\footnote{The missing (symmetric) submatrix of $\Sigma_{b'}$ is the identity matrix added with the entries $\langle B_i,B_j\rangle$.
The missing (symmetric) submatrix of $\Sigma_{a'}^{-1}$ is the identity matrix added with the inner products of the rows of $A$.}.
Let $J=\Sigma_{a'}^{-1}\Sigma_{b'}=
\begin{bmatrix}
\bullet & X_{m'\times d}\\
\bullet & \bullet
\end{bmatrix}
$.
It can be checked that $X_{ij}=(B_{i}(j)-A_{i}(j))$ for every $1\le i\le m'$ and $m'+1\le j \le m'+d$.
Now for every $i$, each of the $\Theta(d)$ places that $A_i$ and $B_i$ differ, $X_{ij}\in\frac{\pm\epsilon}{\sqrt{md}}$.
Hence, their total contribution in $||J-I||_F^2= \Omega(\epsilon^2)$.  
\end{proof}
\begin{proof}[Proof of \cref{thm:kl-lowerbound}]
Consider any algorithm which learns a random distribution $P=\gauss(0,\Sigma)$ from $\mathcal{P}_{\mathcal{C}}$ in $\tv$ distance $c_3\epsilon$, for a small enough constant $c_3$.
Let the learnt distribution be $\wh{P}=\gauss(0,\wh{\Sigma})$.
Then, from Fact~\ref{fact:tvgauss}, and triangle inequality of $\tv$, only the unique distribution $P$ with $\Sigma'=\Sigma$ would satisfy $||\wh{\Sigma}^{-1}\Sigma'-I||_F\le c_4\epsilon$ for every covariance matrix $\Sigma'$ from $\cP_{\cC}$, for an appropriate choice of $c_4$.
This would reveal the random distribution, hence the lower-bound follows.
\end{proof}

\section{Experiments}
\label{sec:experiments}

\paragraph{General Setup}

For our experiments, we explored both polytree networks (generated using random Pr\"{u}fer sequences) and $G(n,p)$ Erd\H{o}s-R\'enyi graphs with bounded \emph{expected} degrees (i.e.\ $p = d/n$ for some bounded degree parameter $d$) using the Python package \texttt{networkx} \cite{hagberg2008exploring}.
Our non-zero edge weights are uniformly drawn from the range $(-2,-1] \cup [1, 2)$.
Once the graph is generated, the linear i.i.d.\ data $X \in \mathbb{R}^{m \times n}$ (with $n$ variables and sample size $m \in \{1000, 2000, \ldots, 5000\}$) is generated by sampling the model $X = B^TX + \eta$, where $\eta \sim \gauss(0, I_{n\times n})$ and $B$ is a strictly upper triangular matrix.\footnote{We do not report the results over the varied variance synthetic data, because their performance are close to the performance of the equal variance synthetic data.}
We report KL divergence between the ground truth and our learned distribution using \cref{eq:decomposition}, averaged over 20 random repetitions.
All experiments were conducted on an Intel Core i7-9750H 2.60GHz CPU.

\paragraph{Algorithms}
The algorithms used in our experiments are as follows: Graphical Lasso \cite{friedman2008sparse}, MLE (empirical) estimator, CLIME \cite{cai2011constrained}, \texttt{LeastSquares}, \texttt{BatchAvgLeastSquares}, \texttt{BatchMedLeastSquares}, \texttt{CauchyEstTree}, and \texttt{CauchyEst}. Specifically, we use \texttt{BatchAvg\_LS+}$x$ and \texttt{BatchMed\_LS+}$x$ to represent the \texttt{\justify BatchAvgLeastSquares} and \texttt{BatchMedLeastSquares} algorithms respectively with a batch size of $p+x$ at each node, where $p$ is the number of parents of that node.

\paragraph{Synthetic data}

\cref{fig:exp_well_rt} compares the KL divergence between the ground truth and our learned distribution over 100 variables between the eight estimators mentioned above.
The first three estimators are for undirected graph structure learning. For this reason, we are not using \cref{eq:decomposition} but the common equation in
\cite[page 13]{duchi2007derivations}
for calculating the KL divergence between multivariate Gaussian distributions.
\cref{fig:all_a} and \cref{fig:all_b} shows the results on ER graphs while \cref{fig:all_c} shows the results for random tree graphs.
The performances of MLE and CLIME are very close to each other, thus are overlapped in \cref{fig:all_a}.
In figure \cref{fig:all_b}, we take a closer look at the results from \cref{fig:all_a} for \texttt{LeastSquares}, \texttt{BatchMedLeastSquares}, 
\texttt{BatchAvgLeastSquares}, \texttt{CauchyEstTree}, and \texttt{CauchyEst} estimators for a clear comparison.
In our experiments, we find that the latter five outperform the GLASSO, CLIME and MLE (empirical) estimators.
With a degree 5 ER graph, \texttt{CauchyEst} performs better than \texttt{CauchyEstTree}, while \texttt{LeastSquares} performs best.
In our experiments for our random tree graphs with in-degree 1 (see \cref{fig:all_c}, we find that the performances between \texttt{CauchyEstTree} and \texttt{CauchyEst} are very close to each other and their plots overlap.

In our experiments, \texttt{CauchyEst} outperforms \texttt{CauchyEstTree} when the $G(n,p)$ graph is generated with a higher degree parameter $d$ (e.g.\ $d > 5$) and the resultant graph is unlikely to yield a polytree structure.

\begin{figure*}
\centering    
\subfigure[Eight algorithms evaluated on ER graph with $d=5$]{\label{fig:all_a}\includegraphics[width=75mm]{\detokenize{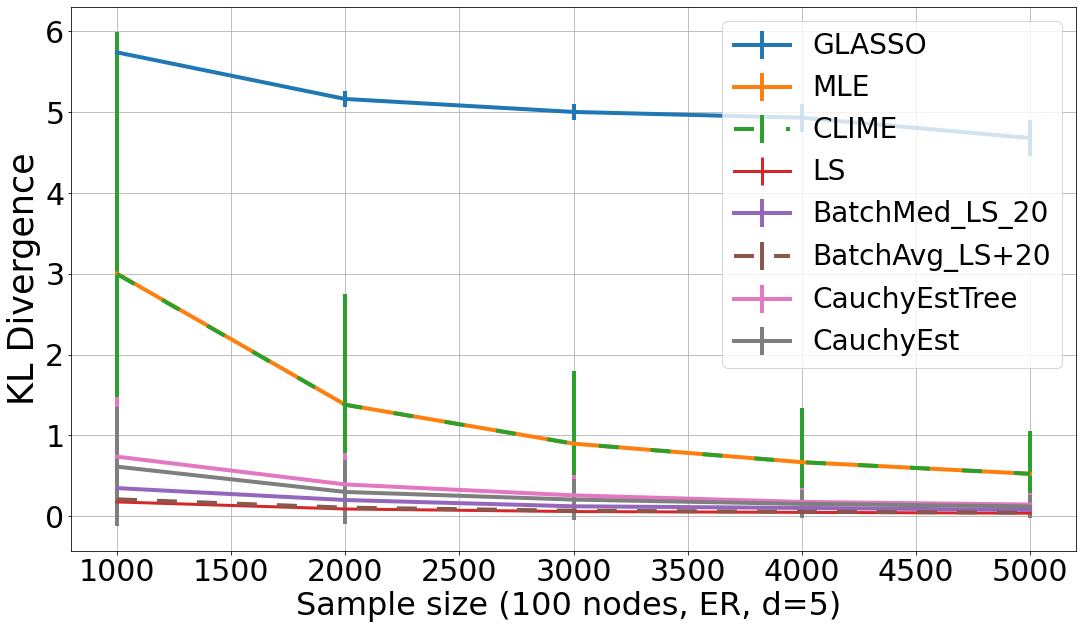}}}
\subfigure[A closer look at some of the algorithms in the plot of \cref{fig:all_a}]{\label{fig:all_b}\includegraphics[width=77mm]{\detokenize{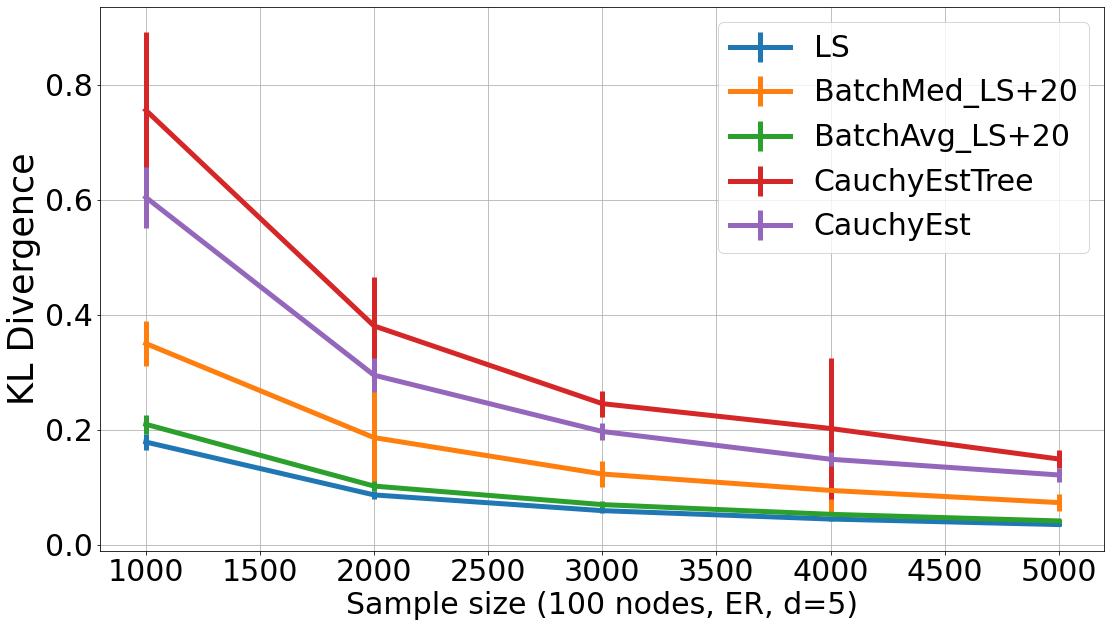}}}
\subfigure[A closer look at some of the algorithms evaluated on random tree graphs]{\label{fig:all_c}\includegraphics[width=77mm]{\detokenize{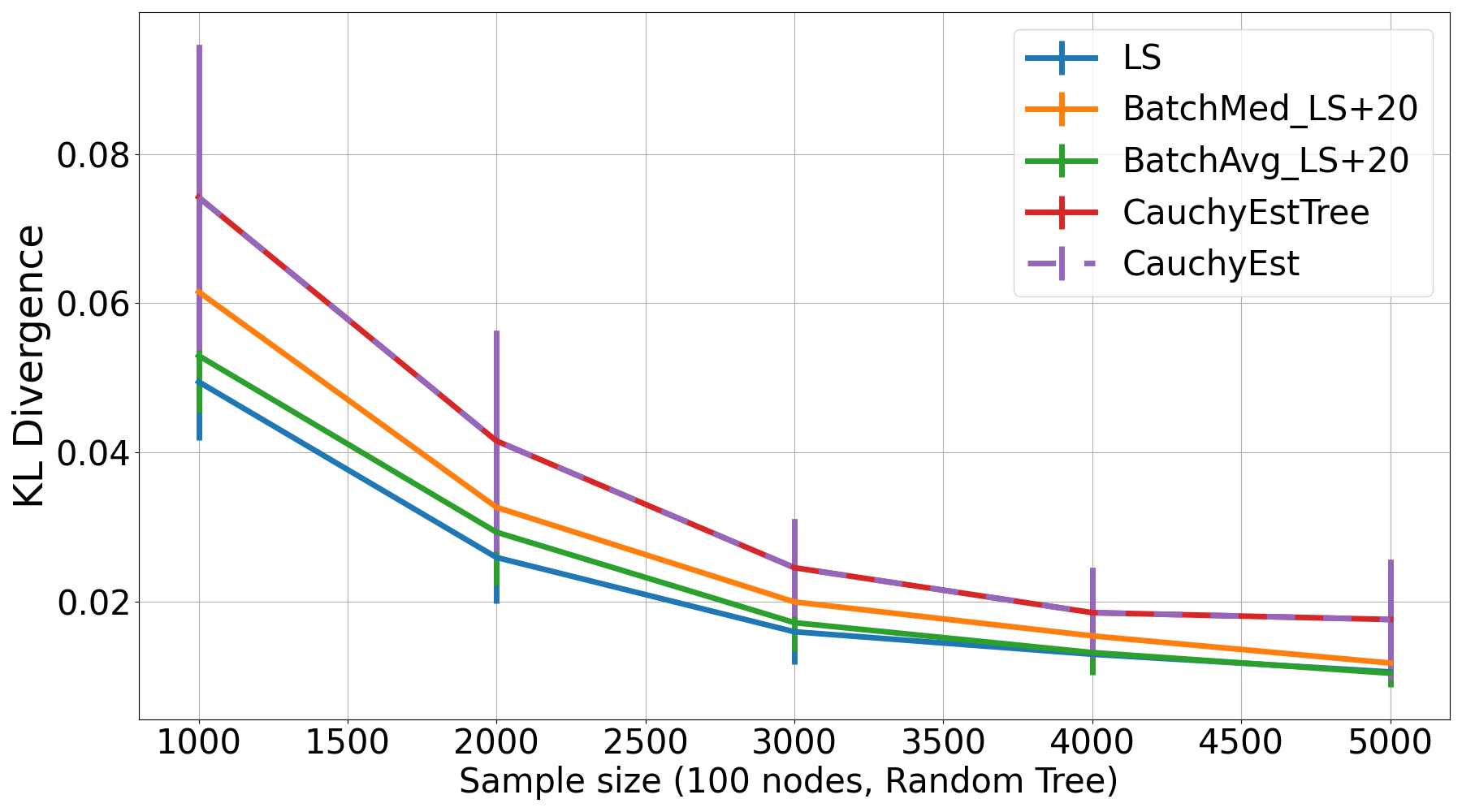}}}
\caption{Experiment on well-conditioned uncontaminated data}
\vskip -0.2cm
\label{fig:exp_well_rt}
\end{figure*}

\paragraph{Real-world datasets}

We also evaluated our algorithms on four real Gaussian Bayesian networks from R package \texttt{bnlearn} \cite{scutari2009learning}.
The \textbf{ECOLI70} graph provided by \cite{schafer2005shrinkage} contains 46 nodes and 70 edges.
The \textbf{MAGIC-NIAB} graph from \cite{scutari2014multiple} contains 44 nodes and 66 edges. 
The \textbf{MAGIC-IRRI} graph contains 64 nodes and 102 edges, and the \textbf{ARTH150} \cite{opgen2007correlation} graph contains 107 nodes and 150 edges.
Experimental results in \cref{fig:bnlearn} show that the error for \texttt{LeastSquares} is smaller than \texttt{CauchyEst} and \texttt{CauchyEstTree} for all the above datasets.

\begin{figure*}
\centering    
\subfigure[Ecoli70 46 nodes]{\label{fig:ecoli70}\includegraphics[width=75mm]{\detokenize{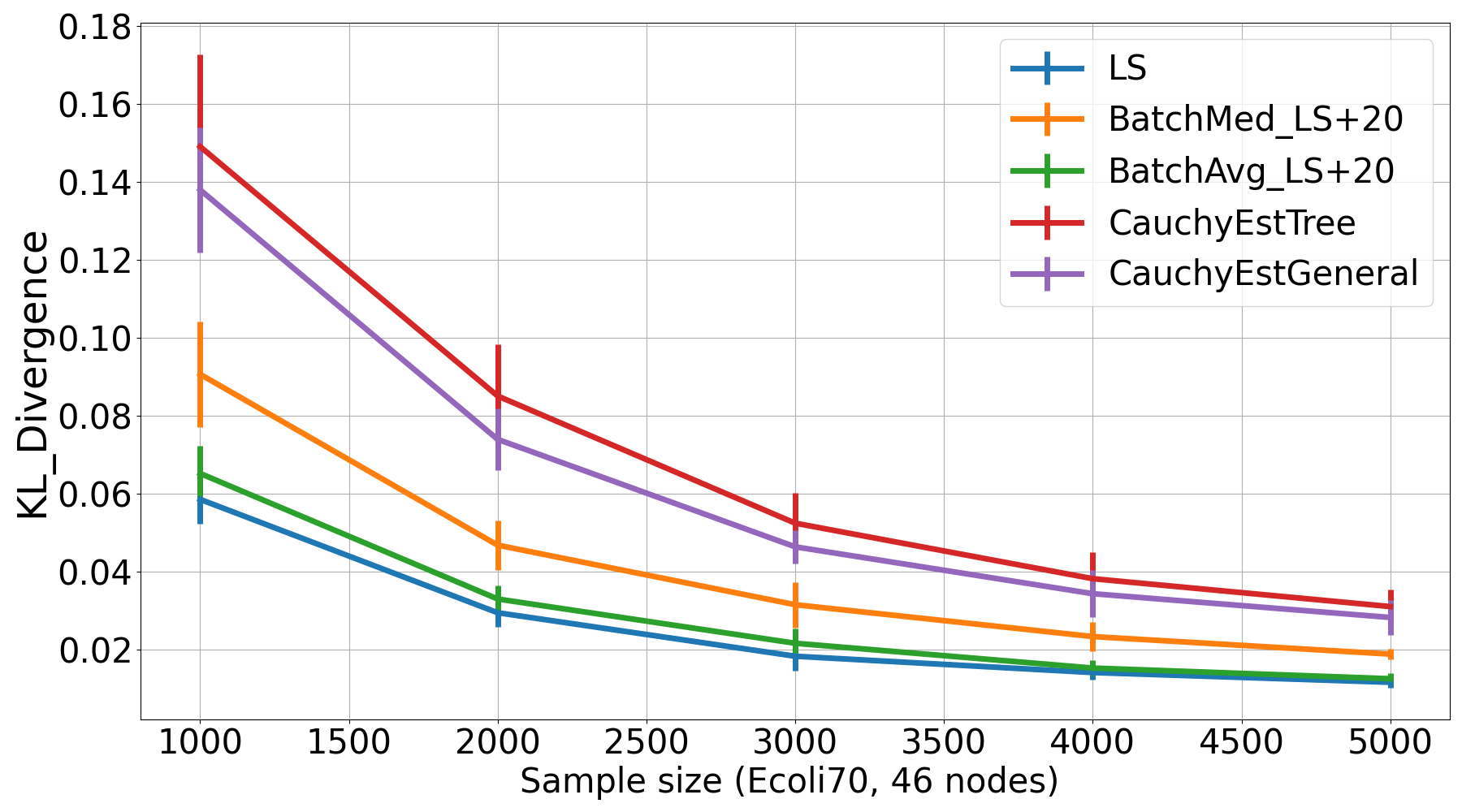}}}
\hspace{2mm}
\subfigure[Magic-niab 44 nodes]{\label{fig:magic-niab}\includegraphics[width=75mm]{\detokenize{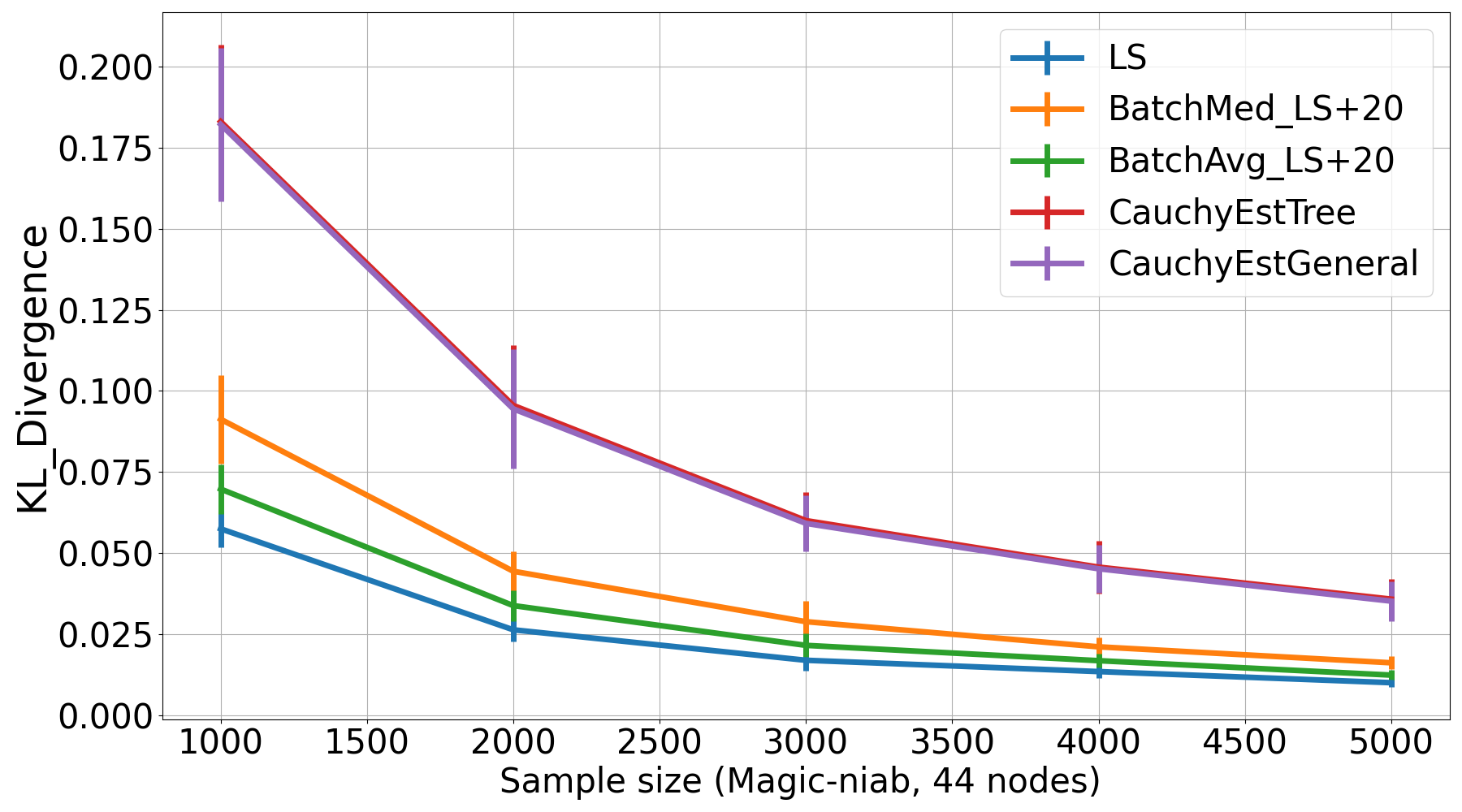}}}
\subfigure[Magic-irri 64  nodes]{\label{fig:magic-irri}\includegraphics[width=75mm]{\detokenize{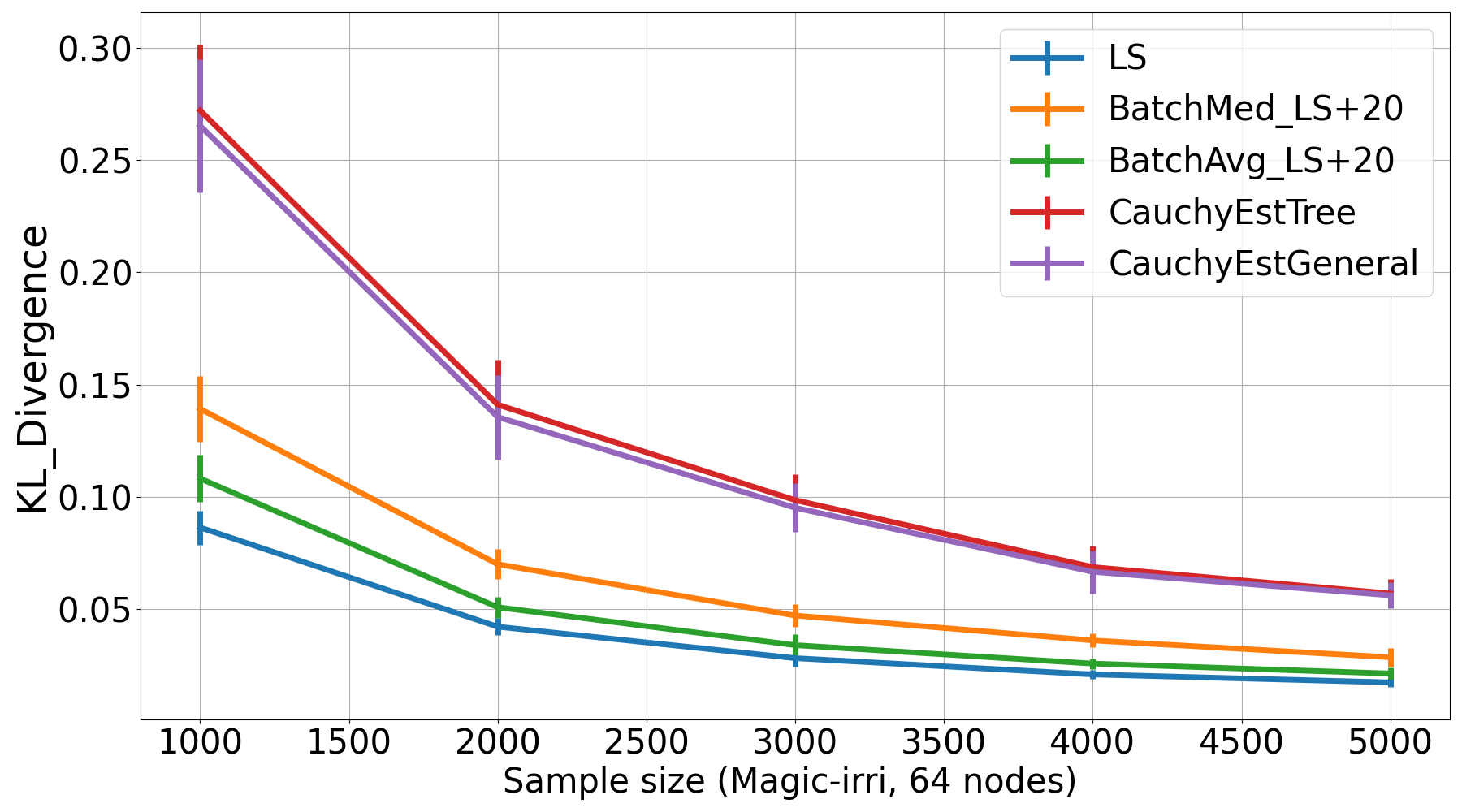}}}
\hspace{2mm}
\subfigure[Arth150 107 nodes]{\label{fig:arth150}\includegraphics[width=75mm]{\detokenize{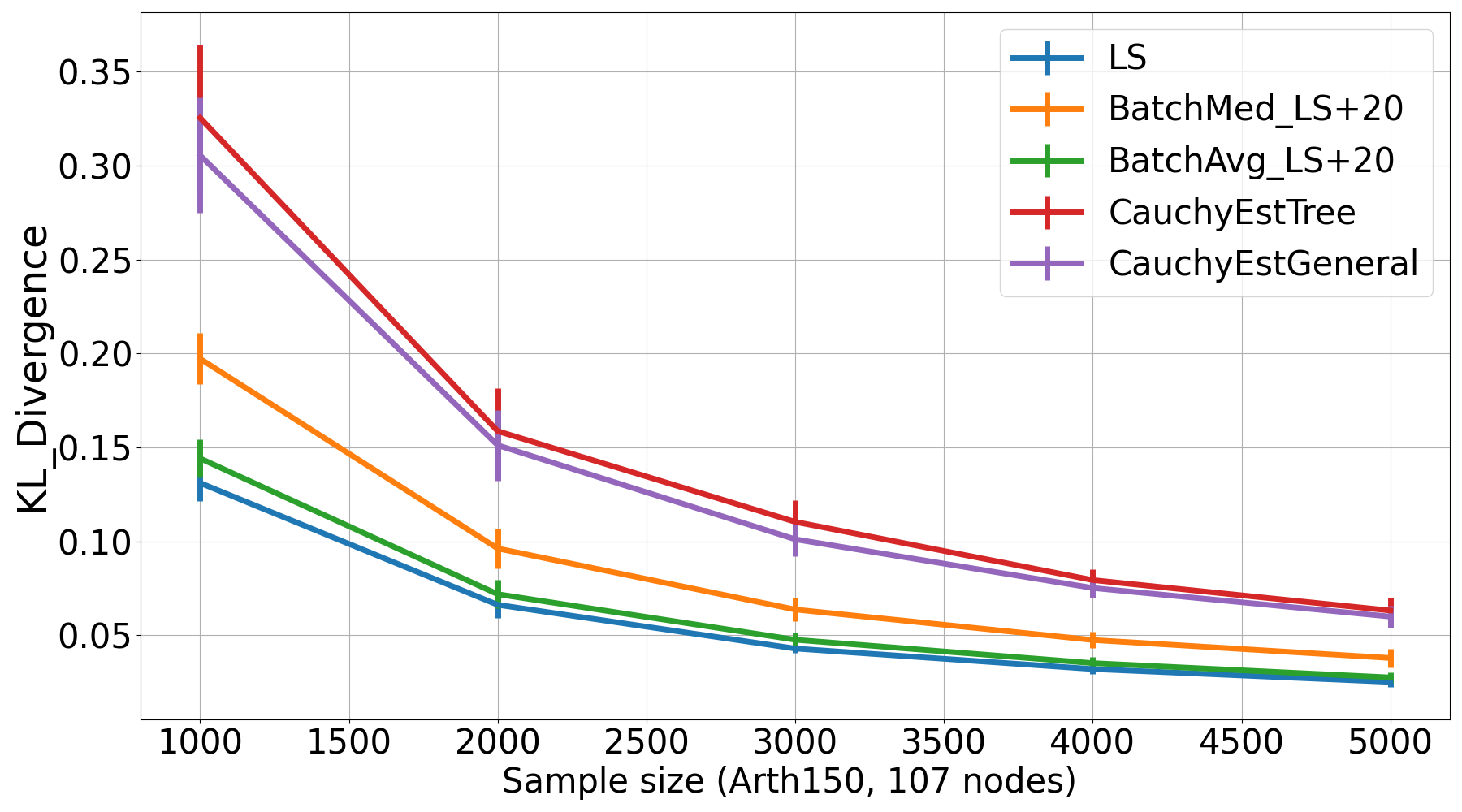}}}
\caption{Experiment results over bnlearn real graph}
\vskip -0.2cm
\label{fig:bnlearn}
\end{figure*}

\paragraph{Contaminated synthetic data}
The contaminated synthetic data is generated in the following way: we randomly choose $5\%$ samples with 5 nodes to be contaminated from the well-conditioned data over $n = 100$ node graphs.
The well-conditioned data has a $\gauss(0,1)$ noise for every variable, while the small proportion of the contaminated data has either $\gauss(1000,1)$ or $\cau(1000,1)$.
In our experiments in \cref{fig:exp_noisy_rt} and \cref{fig:exp_noisy_er}, \texttt{CauchyEst}, \texttt{CauchyEstTree}, and \texttt{BatchMedLeastSquares} outperform \texttt{BatchAvgLeastSquares} and \texttt{LeastSquares} by a large margin. With more than 1000 samples, \texttt{BatchMedLeastSquares} with a batch size of $p+20$ performs similar to \texttt{CauchyEst} and \texttt{CauchyEstTree}, but performs worse with less samples. 
When comparing the performance between \texttt{LeastSqures} and \texttt{\justify BatchAvgLeastSquares} over either a random tree or a ER graph, the experiment in \cref{fig:noisy_rt_a} based on a random tree graph shows that \texttt{LeastSqures} performs better than  \texttt{BatchAvgLeastSquares} when sample size is smaller than 2000, while \texttt{BatchAvgLeastSquares} performs relatively better with more samples. 
Experiment results in \cref{fig:noisy_b} based on ER degree 5 graphs is slightly different from \cref{fig:noisy_rt_a}. In \cref{fig:noisy_b}, \texttt{BatchAvgLeastSquares} performs better than \texttt{LeastSqures} by a large margin. 
Besides, we can also observe that the performances of \texttt{CauchyEst}, \texttt{CauchyEstTree}, and \texttt{BatchMedLeastSquares} are better than the above two estimators and are consistent over different types of graphs. 
For all five algorithms, we use the median absolute devation for robust variance recovery~\cite{huber2004robust} in the contaminated case only (see \cref{algo:mad} in Appendix).

This is because both \texttt{LeastSquares} and \texttt{BatchAvgLeastSquares} use the sample covariance (of the entire dataset or its batches) in the coefficient estimators for the unknown distribution. The presence of a small proportion of outliers in the data can have a large distorting influence on the sample covariance, making them sensitive to atypical observations in the data. On the contrary, our \texttt{CauchyEstTree} and \texttt{BatchMedLeastSquares} estimators are developed using the sample median and hence are resistant to outliers in the data. 

\begin{figure*}
\centering    
\subfigure[all algorithms]{\label{fig:noisy_rt_a}\includegraphics[width=75mm]{\detokenize{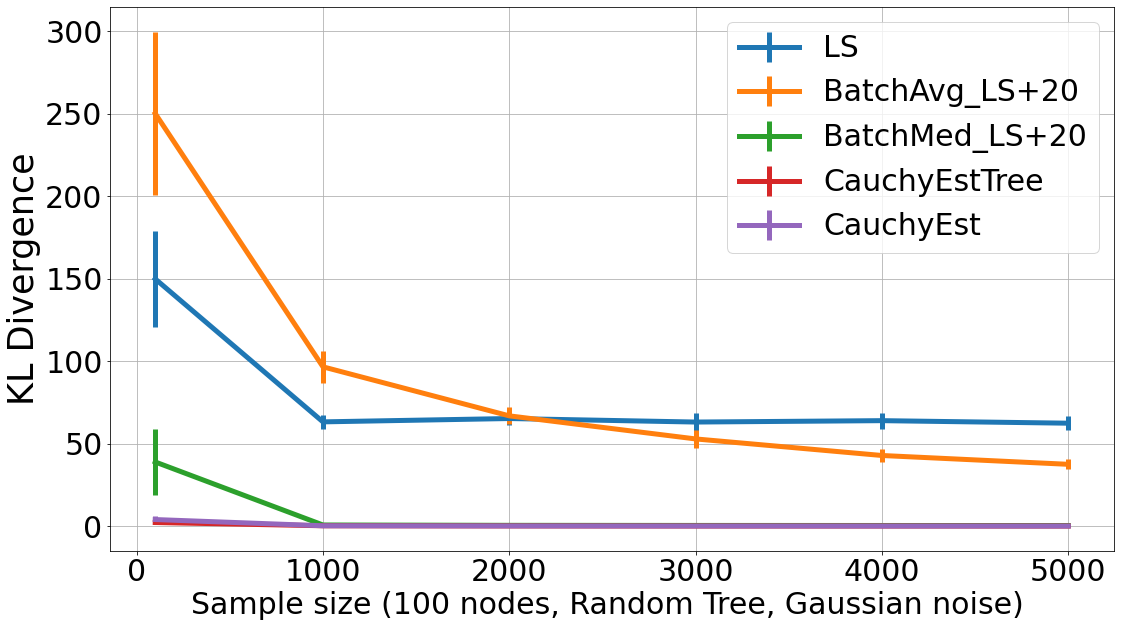}}}
\hspace{2mm}
\subfigure[BatchMed\_LS, CauchyEst, and CauchyEstTree]{\label{fig:noisy_rt_b}\includegraphics[width=75mm]{\detokenize{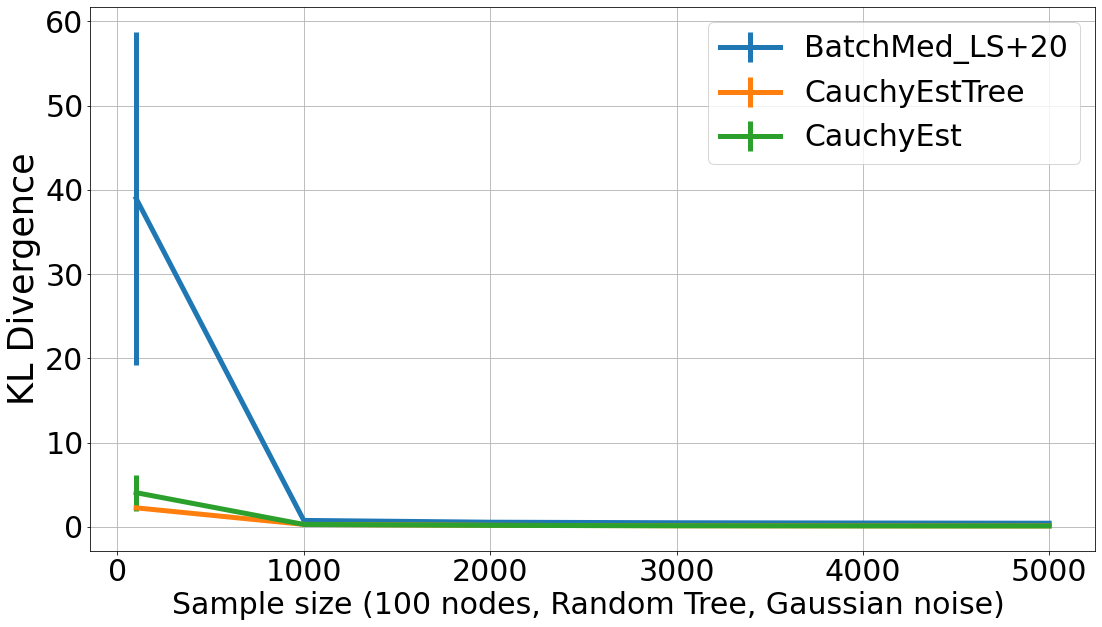}}}
\caption{Experiment results on contaminated data (random tree with Gaussian noise)}
\label{fig:exp_noisy_rt}
\end{figure*}

\begin{figure*}
\centering    
\subfigure[all algorithms]{\label{fig:noisy_b}\includegraphics[width=75mm]{\detokenize{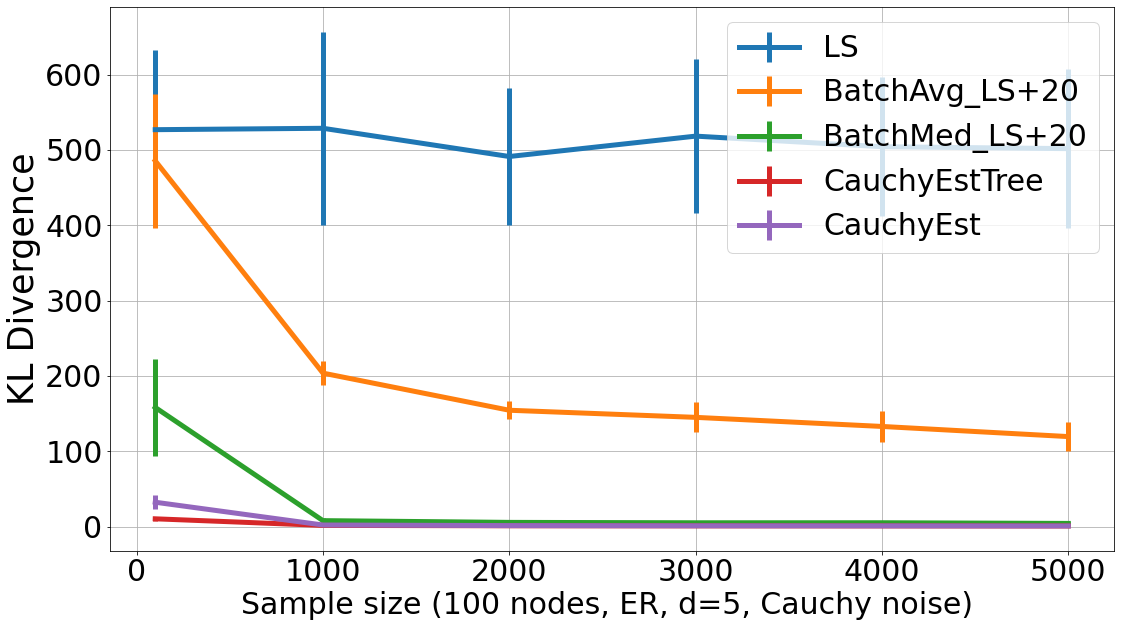}}}
\hspace{2mm}
\subfigure[BatchMed\_LS, CauchyEst, and CauchyEstTree]{\label{fig:noisy_c}\includegraphics[width=75mm]{\detokenize{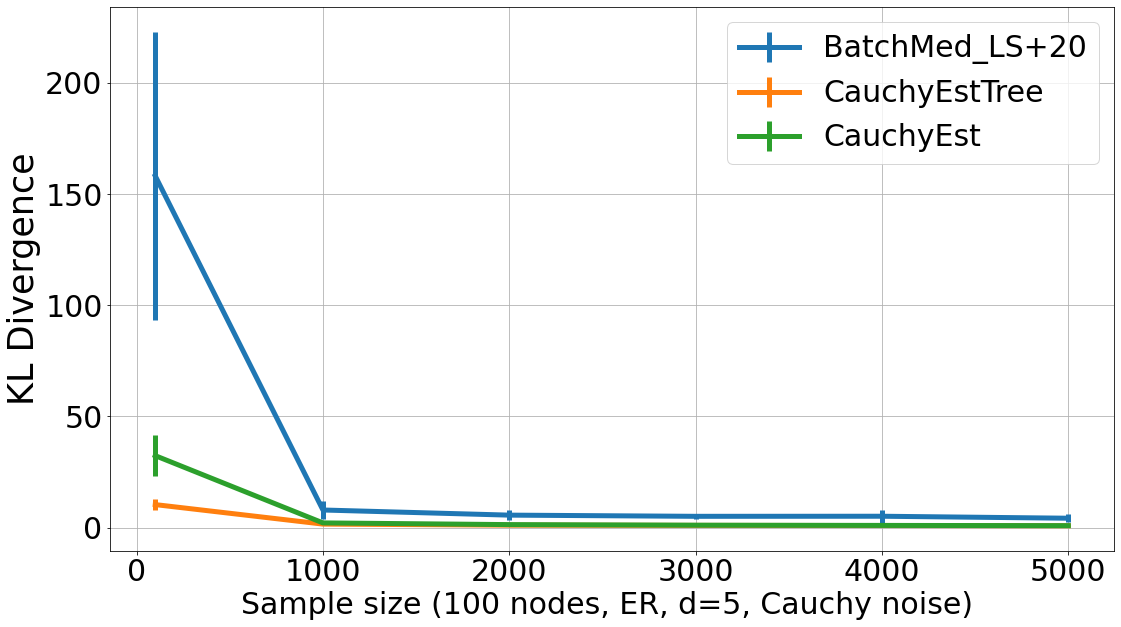}}}
\caption{Experiment results on contaminated data (ER graph with Cauchy noise)}
\vskip -0.2cm
\label{fig:exp_noisy_er}
\end{figure*}

\paragraph{Contaminated real-world datasets}
To test the robustness of the real data in the contaminated condition, we manually contaminate $5\%$ samples of 5 nodes from observational data in \textbf{ECOLI70} and \textbf{ARTH150}. The results are reported in \cref{fig:ecoli70_noisy} and \cref{fig:arth_noisy}. In our experiments, \texttt{CauchyEst} and \texttt{CauchyEstTree} outperforms \texttt{BatchAvgLeastSquares} and \texttt{LeastSquares} by a large margin, and therefore are stable in both contaminated and well-conditioned case. Besides, note that different from the well-conditioned case, here \texttt{CauchyEstTree} performs slightly better than \texttt{CauchyEst}. This is because the Cholesky decomposition used in \texttt{CauchyEst} estimator takes contaminated-data into account. 

\begin{figure*}
\centering    
\subfigure[Ecoli70, 5/46 noisy nodes]{\label{fig:ecoli70_noisy1}\includegraphics[width=75mm]{\detokenize{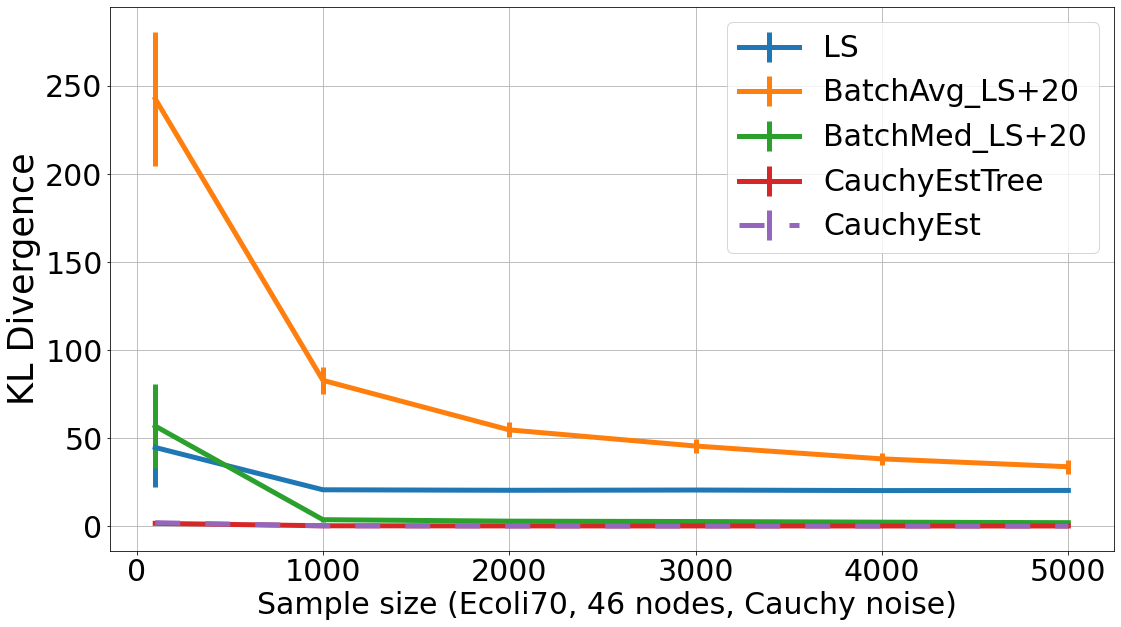}}}
\hspace{2mm}
\subfigure[CauchyEst and CauchyEstTree]{\label{fig:ecoli70_screen}\includegraphics[width=75mm]{\detokenize{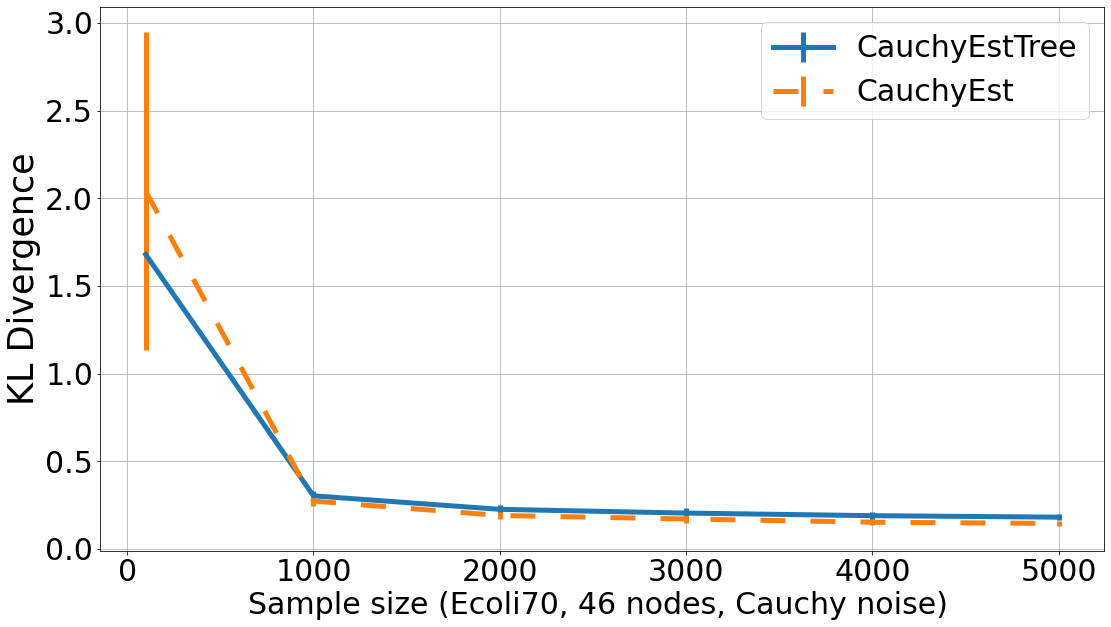}}}
\caption{Ecoli70 under contaminated condition}
\vskip -0.2cm
\label{fig:ecoli70_noisy}
\end{figure*}

\begin{figure*}
\centering    
\subfigure[Arth150, 5/107 noisy nodes]{\label{fig:arth150_noisy}\includegraphics[width=75mm]{\detokenize{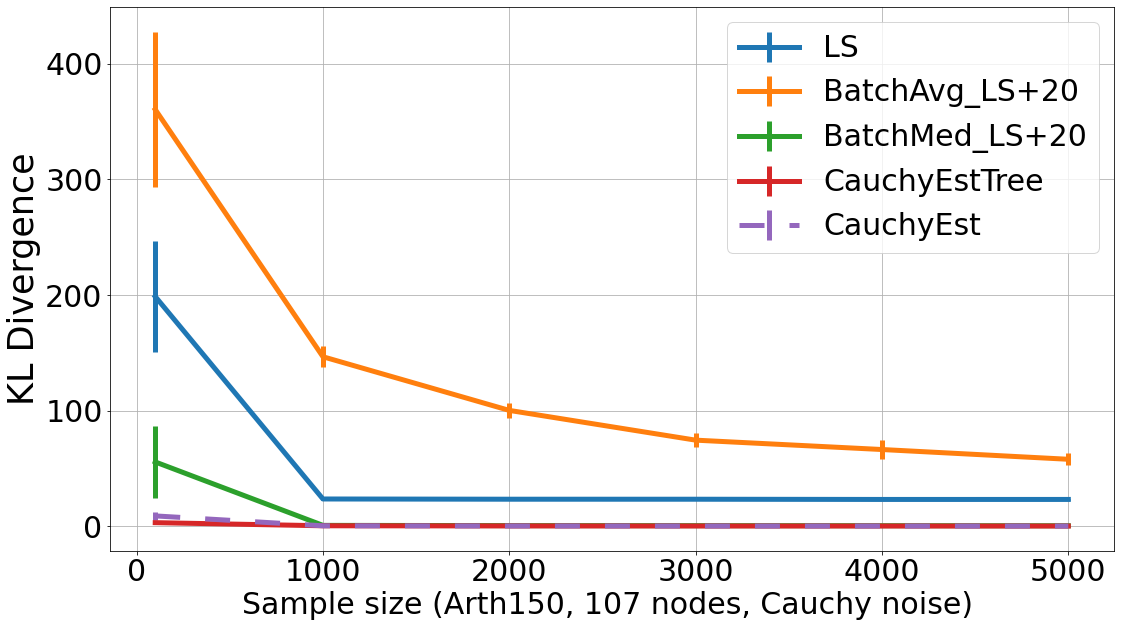}}}
\hspace{2mm}
\subfigure[CauchyEst and CauchyEstTree]{\label{fig:arth150_screen}\includegraphics[width=75mm]{\detokenize{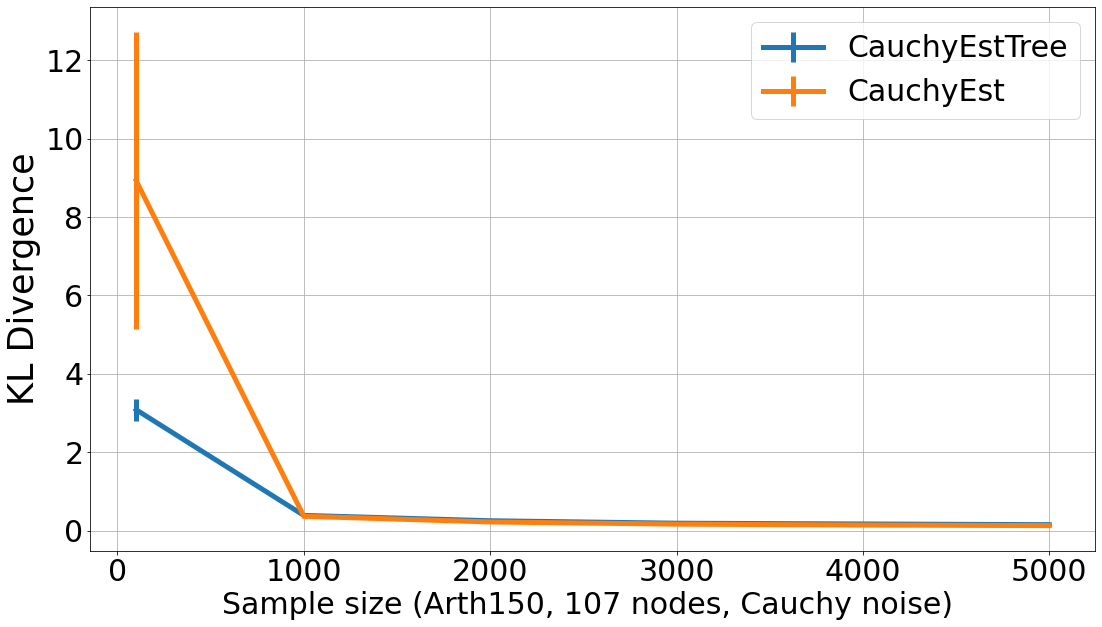}}}
\caption{Arth150 under contaminated condition}
\vskip -0.2cm
\label{fig:arth_noisy}
\end{figure*}

\paragraph{Ill-conditioned synthetic data}
The ill-conditioned data is generated in the following way: we classify the node sets $V$ into either well-conditioned or ill-conditioned nodes. The well-conditioned nodes have a $ \gauss(0, 1)$ noise, while
ill-conditioned nodes have a $ \gauss(0, 10^{-20})$ noise. In our experiments, we choose 3 ill-conditioned nodes over 100 nodes. Synthetic data is sampled from either a random tree or a Erd\H{o}s R\'enyi (ER) model with an expected number of neighbors $d = 5$. Experiments over ill-conditioned Gaussian Bayesian networks through 20 random repetitions are presented in \cref{fig:exp_ill}. For the ill-conditioned settings, we sometimes run into numerical issues when computing the Cholesky decomposition of the empirical covariance matrix $\hat{M}$ in our \texttt{CauchyEst} estimator. 
Thus, we only show the comparision results between \texttt{LeastSquares}, \texttt{BatchAvgLeastSquares}, \texttt{BatchMedLeastSquares}, and \texttt{CauchyEstTree}. 
Here also, the error for \texttt{LeastSquares} decreases faster than the other three estimators. The performance of \texttt{BatchMedLeastSquares} is worse than \texttt{BatchAvgLeastSquares} but slightly better than \texttt{CauchyEstTree} estimator.

\begin{figure*}
\centering    
\subfigure[ER graph, $d=5$]{\label{fig:ill_a}\includegraphics[width=80mm]{\detokenize{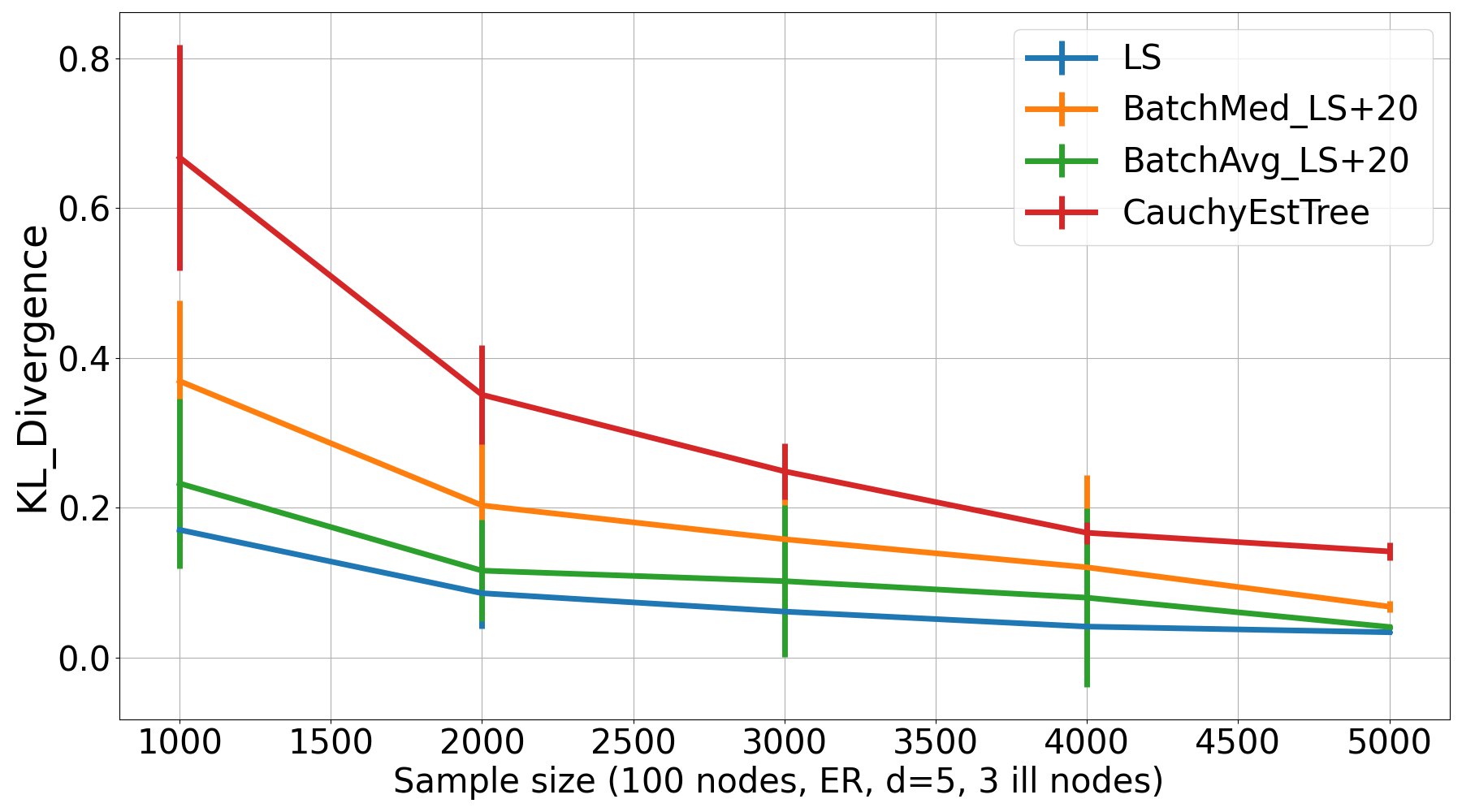}}}
\hspace{2mm}
\subfigure[Random tree]{\label{fig:ill_b}\includegraphics[width=80mm]{\detokenize{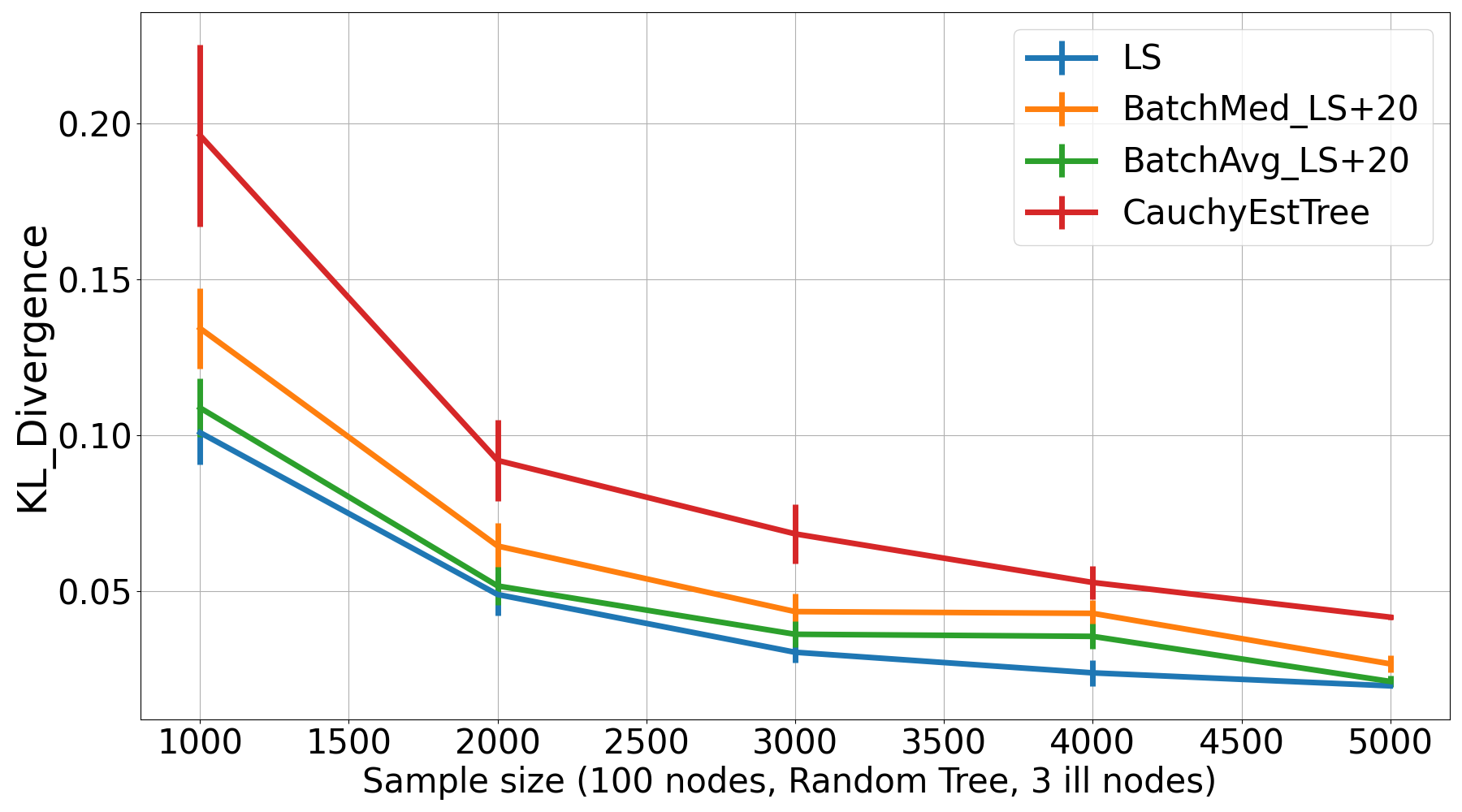}}}
\caption{Experiment results on ill-conditioned data}
\vskip -0.2cm
\label{fig:exp_ill}
\end{figure*}

\paragraph{Agnostic Learning} 
Our theoretical results treat the case that the data is realized by a distribution consistent with the given DAG. In this section, we explore learning of non-realizable inputs, so that there is a non-zero KL divergence between the input distribution $P$ and any distribution consistent with the given DAG.

We conduct agnostic learning experiments by fitting a random sub-graph of the ground truth graph. To do this, we first generate a 100-node ground truth graph $G$, either a random tree with 4 random edges removed or a random ER graph with 9 random edges removed. Our algorithm will try to fit the data from the original Bayes net on $G$ on the edited graph $G^*$. \cref{fig:exp_agn} reports the KL divergence learned over our five estimators. 
We find that \texttt{BatchAvgLeastSquares} estimator performs slightly better than all other estimators in both cases.

\begin{figure*}
\centering    
\subfigure[Random tree: 4 edges removed]{\label{fig:agn_a}\includegraphics[width=80mm]{\detokenize{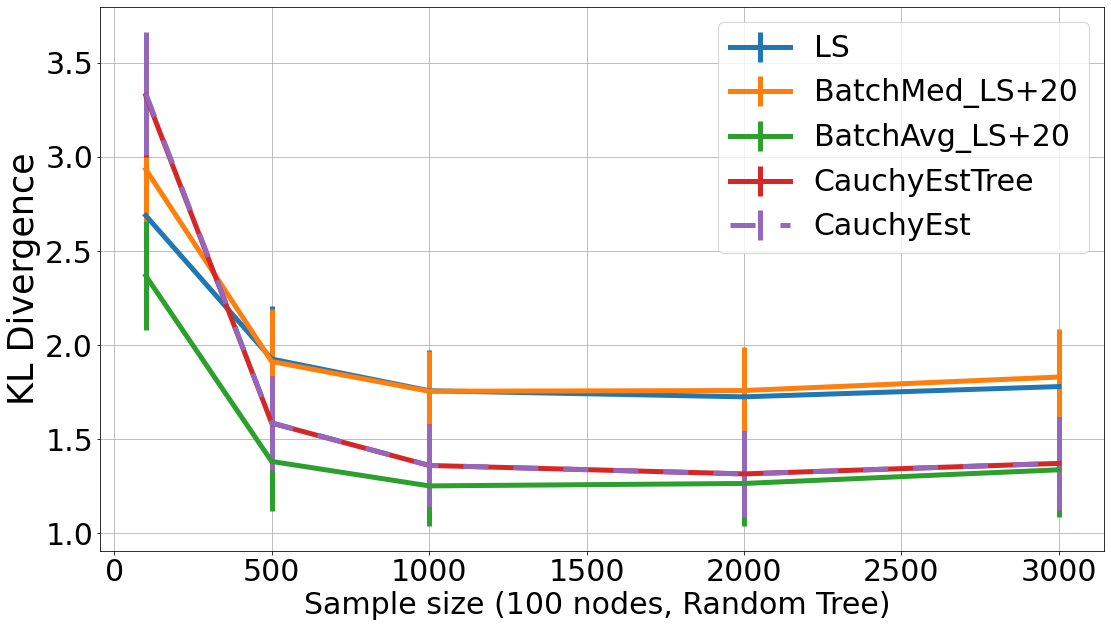}}}
\hspace{2mm}
\subfigure[ER graph, d=5: 9 edges removed]{\label{fig:agn_b}\includegraphics[width=80mm]{\detokenize{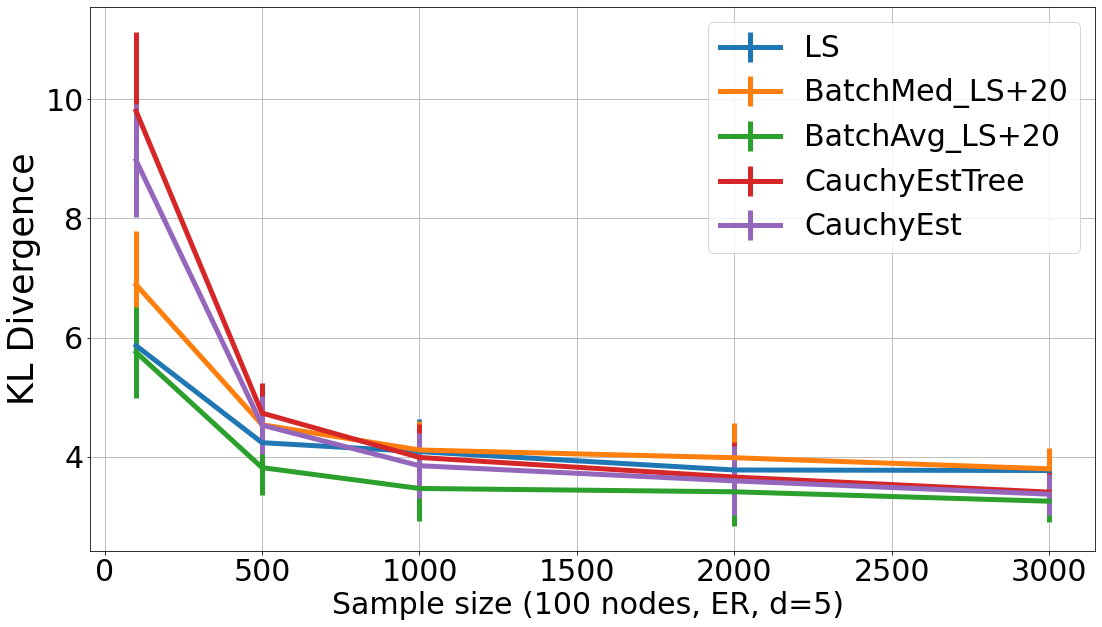}}}
\caption{Agnostic learning}
\vskip -0.2cm
\label{fig:exp_agn}
\end{figure*}

\paragraph{Effect of changing batch size}

\begin{figure*}
\centering    
\subfigure[ER graph, $d=2$]{\label{fig:mean_er_a}\includegraphics[width=80mm]{\detokenize{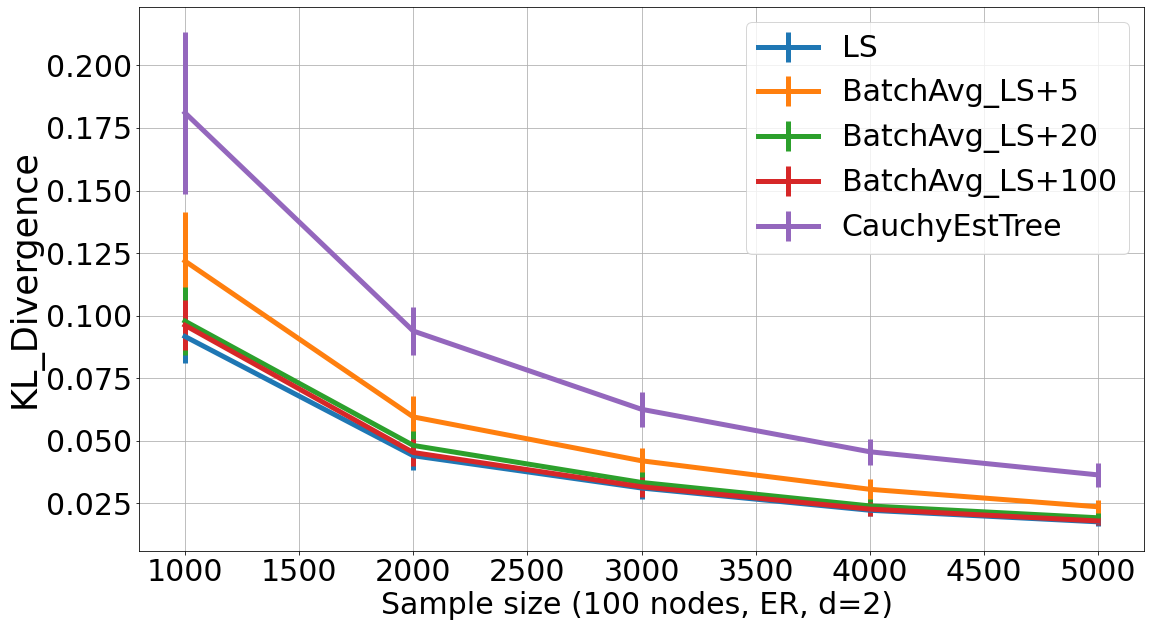}}}
\hspace{2mm}
\subfigure[ER graph, $d=5$]{\label{fig:mean_er_b}\includegraphics[width=80mm]{\detokenize{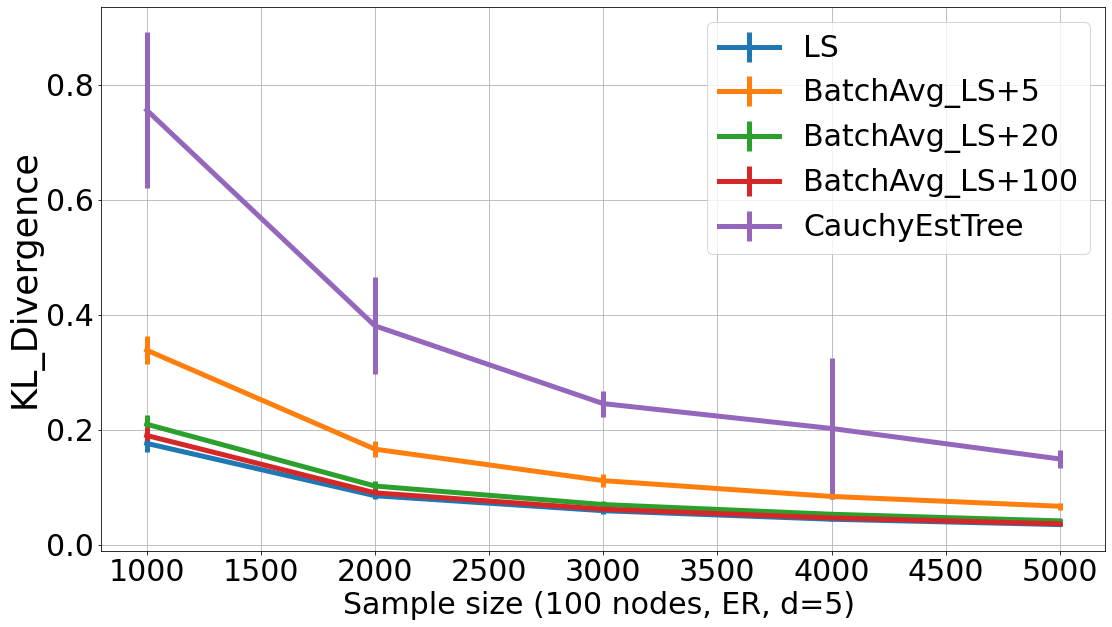}}}
\caption{Effect of changing batch size over Batch Average}
\vskip -0.2cm
\label{fig:batch_mean}
\end{figure*}

\begin{figure*}
\centering    
\subfigure[ER graph, $d=2$]{\label{fig:median_er_a}\includegraphics[width=80mm]{\detokenize{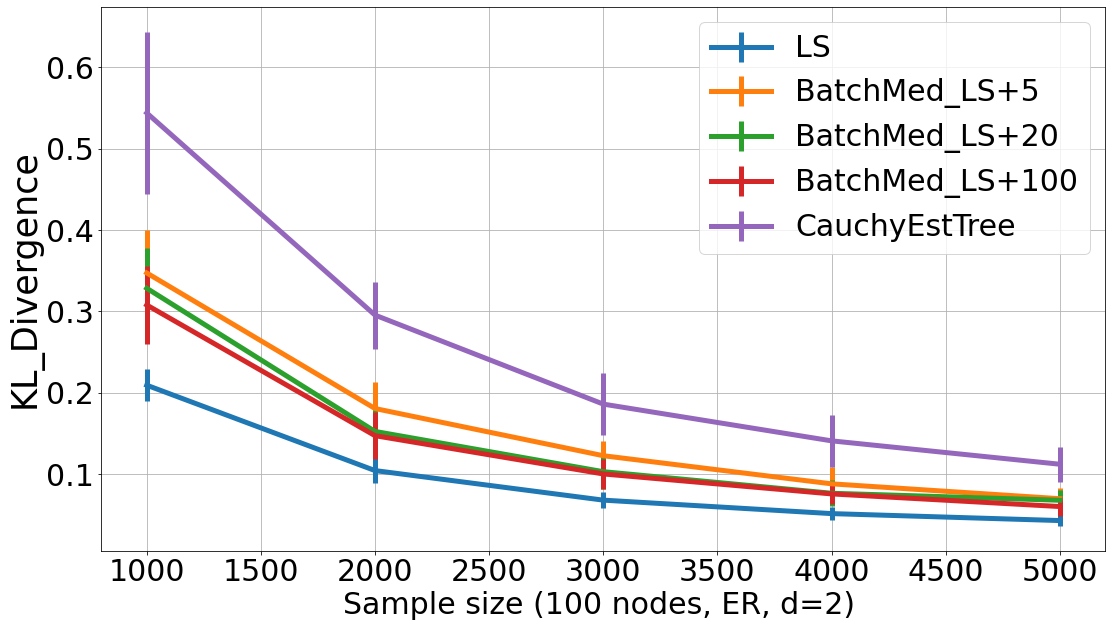}}}
\hspace{2mm}
\subfigure[ER graph, $d=5$]{\label{fig:median_er_b}\includegraphics[width=80mm]{\detokenize{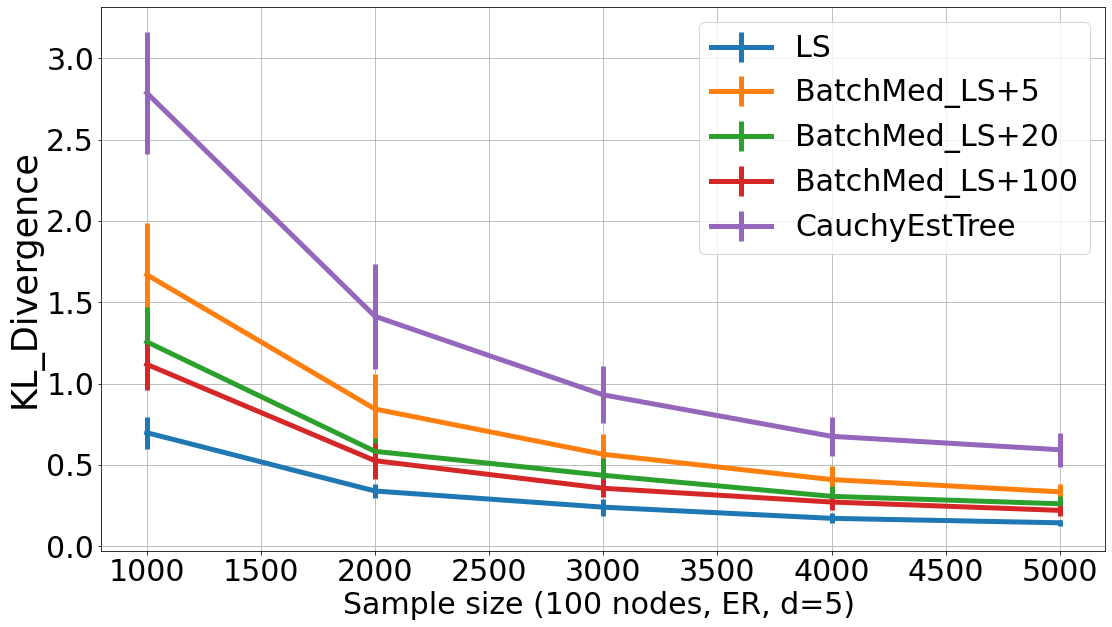}}}
\caption{Effect of changing batch size over Batch Median (ER)}
\vskip -0.2cm
\label{fig:batch_median}
\end{figure*}

Next, we experiment the trade off between the batch-size (eg. batch size = [5, 20, 100]\footnote{We use batch size up to 100 to ensure there are enough batch size from simulation data, so that either mean and median can converge.}) and the KL-divergence of our \texttt{BatchAvgLeastSquares} and \texttt{\justify BatchMedLeastSquares} estimators in detail. As shown in \cref{fig:batch_mean} and \cref{fig:batch_median}, when batch size increases, the results are closer to the \texttt{LeastSquares} estimator. In other words, \texttt{LeastSquares} can be seen as a special case of either \texttt{BatchAvgLeastSquares} or \texttt{BatchMedLeastSquares} with one batch only. Thus, when batch size increases, the performances of \texttt{BatchAvgLeastSquares} and \texttt{BatchMedLeastSquares} are closer to the \texttt{LeastSquares} estimator. On the contrary, the \texttt{CauchyEstTree} estimator can be seen as the estimator with a batch size of $p$. Therefore, with smaller batch size (eg. batch size $= p+5$),  \texttt{BatchAvgLeastSquares} and \texttt{BatchMedLeastSquares} performs closer to the \texttt{CauchyEstTree} estimator. 

\paragraph{Runtime comparison}
We now give the amount of time spent by each algorithm to process a degree-5 ER graph on 100 nodes with 500 samples.
\texttt{LeastSquares} algorithm takes 0.0096 seconds, \texttt{BatchAvgLeastSquares} with a batch size of $p+20$ takes 0.0415 seconds, \texttt{BatchMedLeastSquares} with a batch size of $p+20$ takes 0.0290 seconds, \texttt{CauchyEstTree} takes 0.6063 seconds, and \texttt{CauchyEst} takes 0.6307 seconds.
The timings given above are representative of the relative running times of these algorithms across different graph sizes.

\paragraph{Takeaways}
In summary, the \texttt{LeastSquares} estimator performs the best among all algorithms on uncontaminated datasets (both real and synthetic) generated from Gaussian Bayesian networks. This holds even when the data is ill-conditioned. However, if a fraction of the samples are contaminated, \texttt{CauchyEst}, \texttt{CauchyEstTree} and \texttt{BatchMedLeastSquares} outperform \texttt{LeastSquares} and \texttt{BatchAvgLeastSquares} by a large margin under different noise and graph types. If the data is not generated according to the input graph (i.e., the non-realizable/agnostic learning setting), then \texttt{BatchAvgLeastSquares}, \texttt{CauchyEst}, and \texttt{CauchyEstTree} have a better tradeoff between the error and sample complexity than the other algorithms, although we do not have a formal explanation.

\subsubsection*{Acknowledgements}
This research/project is supported by the National Research Foundation, Singapore under its AI Singapore Programme (AISG Award No: AISG-PhD/2021-08-013).

\bibliographystyle{alpha}
\bibliography{refs}
\newpage
\appendix

\section{Details on decomposition of KL divergence}
\label{sec:decomposition-details}

In this section, we provide the full derivation of \cref{eq:decomposition}.

For notational convenience, we write $x$ to mean $(x_1, \ldots, x_n)$, $\pi_i(x)$ to mean the values given to parents of variable $X_i$ by $x$, and $\mathcal{P}(x)$ to mean $\mathcal{P}(X_1 = x_1, \ldots, X_n = x_n)$.
Observe that
\begin{align*}
&\; \kl(\mathcal{P}, \mathcal{Q})\\
=&\; \int_{x} \mathcal{P}(x) \log \Paren{\frac{\mathcal{P}(x)}{\mathcal{Q}(x)}} dx\\
=&\; \int_{x} \mathcal{P}(x) \log \Paren{\frac{\Pi_{i=1}^n \mathcal{P}(x_i \mid \pi_i(x))}{\Pi_{i=1}^n \mathcal{Q}(x_i \mid \pi_i(x))}} dx && \text{$(\star)$}\\
=&\; \sum_{i=1}^n \int_{x} \mathcal{P}(x) \log
\Paren{\frac{\mathcal{P}(x_i \mid \pi_i(x))}{\mathcal{Q}(x_i \mid \pi_i(x))}} dx\\
=&\; \sum_{i=1}^n \int_{x_i, \pi_i(x)} \mathcal{P}(x_i, \pi_i(x)) \log
\Paren{\frac{\mathcal{P}(x_i \mid \pi_i(x))}{\mathcal{Q}(x_i \mid \pi_i(x))}} d x_i d \pi_i(x) && \text{Marginalization}
\end{align*}
where $(\star)$ is due to the Bayesian network decomposition of joint probabilities.
Let us define
\[
\cp(\alpha^*_i, \wh{\alpha}_i)
= \int_{x_i, \pi_i(x)} \mathcal{P}(x_i, \pi_i(x)) \log
\Paren{\frac{\mathcal{P}(x_i \mid \pi_i(x))}{\mathcal{Q}(x_i \mid \pi_i(x))}} d x_i d \pi_i(x)    
\]
where each $\wh{\alpha}_i$ and $\alpha^*_i$ represent the parameters that relevant to variable $X_i$ from $\wh{\alpha}$ and $\alpha^*$ respectively.
Under this notation, we can write
$
\kl(\mathcal{P}, \mathcal{Q})
= \sum_{i=1}^n \cp(\alpha^*_i, \wh{\alpha}_i)
$.

Fix a variable of interest $Y$ with parents $X_1, \ldots, X_p$, each with coefficient $c_i$, and variance $\sigma^2$.
That is, $Y = \eta_y + \sum_{i=1}^p c_i X_i$ for some $\eta_y \sim N(0, \sigma^2)$ that is independent of $X_1, \ldots, X_p$.
By denoting $X = x$ (i.e.\ $X_1 = x_1, \ldots, X_p = x_p$) and $c = (c_1, \ldots, c_p)$, we can write the conditional distribution density of $Y$ as
\[
\Pr\Paren{y \mid x, c, \sigma}
= \frac{1}{\sigma \sqrt{2 \pi}} \exp \Paren{-\frac{1}{2 \sigma^2} \cdot \Paren{y - \sum_{i=1}^p c_i X_i}^2}
\]

We now analyze $\cp(\alpha^*_y, \wh{\alpha}_y)$ with respect to the our estimates $\wh{\alpha}_y = (\wh{A}, \wh{\sigma}_y)$ and the hidden true parameters $\alpha^*_y = (A, \sigma_y)$, where $\wh{A} = (\wh{a}_{y \leftarrow 1}, \ldots, \wh{a}_{y \leftarrow p})$ and $A = (a_{y \leftarrow 1}, \ldots, a_{y \leftarrow p})$.

With respect to variable $Y$, we see that
\begin{align*}
&\; \cp(\alpha^*_y, \wh{\alpha}_y)\\
= &\; \int_{x, y} \mathcal{P}(x, y) \ln \Paren{\frac{\frac{1}{\sigma_y \sqrt{2 \pi}} \exp \Paren{-\frac{1}{2 \sigma^2} \cdot \Paren{y - \sum_{i=1}^p a_{y \leftarrow i} X_i}^2}}{\frac{1}{\wh{\sigma}_y \sqrt{2 \pi}} \exp \Paren{-\frac{1}{2 \wh{\sigma}_y^2} \cdot \Paren{y - \sum_{i=1}^p \wh{a}_{y \leftarrow i} X_i}^2}}} dy \ dx\\
= &\; \ln \Paren{\frac{\wh{\sigma}_y}{\sigma_y}} - \frac{1}{2 \sigma_y^2} \cdot \E_{x,y} \Paren{y - \sum_{i=1}^p a_{y \leftarrow i} X_i}^2 + \frac{1}{2 \wh{\sigma}_y^2} \cdot \E_{x,y} \Paren{y - \sum_{i=1}^p \wh{a}_{y \leftarrow i} X_i}^2\\
= &\; \ln \Paren{\frac{\wh{\sigma}_y}{\sigma_y}} - \frac{1}{2 \sigma_y^2} \cdot \E_{x,y} \Paren{y - A^\top X}^2 + \frac{1}{2 \wh{\sigma}_y^2} \cdot \E_{x,y} \Paren{y - \wh{A}^\top X}^2\\
\end{align*}

By defining $\Delta = \wh{A} - A$, we can see that for any instantiation of $y, X_1, \ldots, X_p$,
\begin{align*}
&\; \Paren{y - \wh{A}^\top X}^2\\
= &\; \Paren{y - (\Delta + A)^\top X}^2 && \text{By definition of $\Delta$}\\
= &\; \Paren{(y - A^\top X) - \Delta^\top X}^2\\
= &\; (y - A^\top X)^2 - 2 (y - A^\top X)(\Delta^\top X) + \Paren{\Delta^\top X}^2\\
= &\; (y - A^\top X)^2 - 2 \Paren{y \Delta^\top X - A^\top X \Delta^\top X} + \Paren{\Delta^\top X}^2\\
= &\; (y - A^\top X)^2 - 2 \Paren{y X^\top \Delta - A^\top X X^\top \Delta} + \Delta^\top XX^\top \Delta && \text{Since $\Delta^\top X$ is just a number}
\end{align*}

Denote the covariance matrix with respect to $X_1, \ldots, X_p$ as $M \in \mathbb{R}^{p \times p}$, where $M_{i,j} = \E\Brac{X_i X_j}$.
Then, we can further simplify $\cp(\alpha^*_y, \wh{\alpha}_y)$ as follows:
\begin{align*}
&\; \cp(\alpha^*_y, \wh{\alpha}_y)\\
= &\; \ln \Paren{\frac{\wh{\sigma}_y}{\sigma_y}} - \frac{1}{2 \sigma_y^2} \cdot \E_{x,y} \Paren{y - A^\top X}^2 + \frac{1}{2 \wh{\sigma}_y^2} \cdot \E_{x,y} \Paren{y - \wh{A}^\top X}^2 && \text{From above}\\
= &\; \ln \Paren{\frac{\wh{\sigma}_y}{\sigma_y}} - \frac{1}{2 \sigma_y^2} \cdot \E_{x,y} \Paren{y - A^\top X}^2\\
&\; + \frac{1}{2 \wh{\sigma}_y^2} \cdot \E_{x,y} \Brac{(y - A^\top X)^2 - 2 \Paren{y X^\top \Delta - A^\top X X^\top \Delta} + \Delta^\top XX^\top \Delta} && \text{From above}\\
= &\; \ln \Paren{\frac{\wh{\sigma}_y}{\sigma_y}} - \frac{1}{2 \sigma_y^2} \cdot \E_{x,y} \eta_y^2 + \frac{1}{2 \wh{\sigma}_y^2} \cdot \E_{x,y} \Brac{\eta_y^2 - 2 \Paren{\eta_y X^\top \Delta} + \Delta^\top XX^\top \Delta} && \text{$(\dag)$}\\
= &\; \ln \Paren{\frac{\wh{\sigma}_y}{\sigma_y}} - \frac{1}{2 \sigma_y^2} \cdot \sigma_y^2 + \frac{1}{2 \wh{\sigma}_y^2} \cdot \Brac{\sigma_y^2 - 0 + \Delta^\top M \Delta} && \text{$(\ast)$}\\
= &\; \ln \Paren{\frac{\wh{\sigma}_y}{\sigma_y}} - \frac{1}{2} + \frac{\sigma_y^2}{2 \wh{\sigma}_y^2} + \frac{\Delta^\top M \Delta}{2 \wh{\sigma}_y^2}\\
= &\; \ln \Paren{\frac{\wh{\sigma}_y}{\sigma_y}} + \frac{\sigma_y^2 - \wh{\sigma}_y^2}{2 \wh{\sigma}_y^2} + \frac{\Delta^\top M \Delta}{2 \wh{\sigma}_y^2}
\end{align*}
where $(\dag)$ is because $y = \eta_y + A^\top X$ while $(\ast)$ is because $\eta_y \sim N(0, \sigma_y^2)$, $\E_{x,y} \Paren{\eta_y X^\top \Delta} = \E_{x,y} \eta_y \cdot \E_{x,y} X^\top \Delta = 0$, and $\E_{x,y} \Delta^\top XX^\top \Delta = \Delta^\top (\E_{x,y} XX^\top) \Delta = \Delta^\top M \Delta$.

In conclusion, we have
\begin{equation}
\kl(\mathcal{P}, \mathcal{Q})
= \sum_{i=1}^n \cp(\alpha^*_i, \wh{\alpha}_i)
= \sum_{i=1}^n \ln \Paren{\frac{\wh{\sigma}_i}{\sigma_i}} + \frac{\sigma_i^2 - \wh{\sigma}_i^2}{2 \wh{\sigma}_i^2} + \frac{\Delta_i^\top M_i \Delta_i}{2 \wh{\sigma}_i^2}
\end{equation}
where $M_i$ is the covariance matrix associated with variable $X_i$, $\alpha_i^* = (A_i, \sigma_i)$ is the coefficients and variance associated with variable $X_i$, $\alpha_i = (\wh{A}_i, \wh{\sigma}_i)$ are the estimates for $\alpha_i^*$, and $\Delta_i = \wh{A}_i - A_i$.

\section{Deferred proofs}
\label{sec:deferred-proofs}

This section provides the formal proofs that were deferred in favor for readability.
For convenience, we will restate the statements before proving them.

The next two lemmata \cref{lem:gtg} and \cref{lem:Getanorm} are used in the proof of \cref{lem:least-squares-single}, which is in turn used in the proof of \cref{thm:two-phased-leastsquares}.
We also use \cref{lem:gtg} in the proof of \cref{lem:BatchAvgLeastSquares-single}.
The proof of \cref{lem:gtg} uses the standard result of \cref{lem:smallest-singular-value}.

\begin{lemma}[\cite{rudelson2009smallest}; Theorem 6.1 and Equation 6.10 in \cite{wainwright2019high}]
\label{lem:smallest-singular-value}
Let $\ell \geq d$ and $G \in \R^{\ell \times d}$ be a matrix with i.i.d.\ $N(0,1)$ entries.
Denote $\sigma_{\min}(G)$ as the smallest singular value of $G$.
Then, for any $0 < t < 1$, we have
$
\Pr\Paren{ \sigma_{\min}(G) \geq \sqrt{\ell}(1-t) - \sqrt{d} }
\leq \exp\Paren{ -\ell t^2/2 }
$.
\end{lemma}

\gtg*
\begin{proof}[Proof of \cref{lem:gtg}]
Observe that $G^\top G$ is symmetric, thus $(G^\top G)^{-1}$ is also symmetric and the eigenvalues of $G^\top G$ equal the singular values of $G^\top G$.
Also, note that event that $G^\top G$ is singular has measure 0.\footnote{Consider fixing all but one arbitrary entry of $G$. The event of this independent $N(0,1)$ entry making $\det(G^\top G) = 0$ has measure 0.}

By definition of operation norm, $\Norm{(G^\top G)^{-1}}$ equals the square root of \emph{maximum} eigenvalue of
\[
((G^\top G)^{-1})^\top ((G^\top G)^{-1}) = ((G^\top G)^{-1})^2,
\]
where the equality is because $(G^\top G)^{-1}$ is symmetric.
Since $G^\top G$ is invertible, we have $\|(G^\top G)^{-1}\| = 1/\|G^\top G\|$, which is equal to the inverse of \emph{minimum} eigenvalue $\lambda_{\min}(G^\top G)$ of $G^\top G$, which is in turn equal to the square of minimum singular value $\sigma_{\min}(G)$ of $G$.

Therefore, the following holds with probability at least $1 - \exp\Paren{ -k c_1^2/2 }$:
\[
\Norm{\Paren{G^\top G}^{-1}}
= \frac{1}{\Norm{G^\top G}}
= \frac{1}{\lambda_{\min}(G^\top G)}
= \frac{1}{\sigma_{\min}^2(G)}
\leq \frac{1}{\Paren{\sqrt{k}(1 - c_1) - \sqrt{d}}^2}
\leq \frac{1}{\Paren{1 - 2c_1}^2 k}
\]
where the second last inequality is due to \cref{lem:smallest-singular-value} and the last inequality holds when $k \geq d / c_1^2$.
\end{proof}

\Getanorm*
\begin{proof}[Proof of \cref{lem:Getanorm}]
Let us denote $g_r \in \mathbb{R}^{k}$ as the $r^{th}$ row of $G^\top$.
Then, we see that $\|G^\top \eta\|_2^2 = \sum_{r=1}^{p} \langle g_r, \eta \rangle^2$.
For any row $r$, we see that $\langle g_r, \eta \rangle = \|\eta\|_2 \cdot \langle g_r, \eta / \|\eta\|_2 \rangle$.
We will bound values of $\|\eta\|_2$ and $| \langle g_r, \eta / \|\eta\|_2 \rangle |$ separately.

It is well-known (e.g.\ see \cite[Lemma 2]{jin2019short}) that the norm of a Gaussian vector concentrates around its mean.
So, $\Pr\Paren{\|\eta\|_2 \leq 2 \sigma\sqrt{k}} \leq 2 \exp\Paren{ - 2k}$.
Since $g_r \sim N(0, I_{k})$ and $\eta$ are independent, we see that $\langle g_r, \eta / \|\eta\|_2 \rangle \sim N(0,1)$.
By standard Gaussian bounds, we have that $\Pr\Brac{ | \langle g_r, \eta / \|\eta\|_2 \rangle | \geq c_2 } \leq \exp\Paren{ - c_2^2 / 2}$.

By applying a union bound over these two events, we see that $\| \langle g_r, \eta \rangle \| < 2 \sigma c_2 \sqrt{k}$ for any row with probability $2 \exp\Paren{ - 2k} + \exp\Paren{ - c_2^2 / 2}$.
The claim follows from applying a union bound over all $p$ rows.
\end{proof}

\cauchymedianconvergence*
\begin{proof}[Proof of \cref{lem:cauchy-median-convergence}]
Let $S_{> \tau} = \sum_{i=1}^m \mathbbm{1}_{X_i > \tau}$ be the number of values that are larger than $\tau$, where $\E[\mathbbm{1}_{X_i > \tau}] = \Pr(X \geq \tau)$.
Similarly, let $S_{< -\tau}$ be the number of values that are smaller than $-\tau$.
If $S_{> \tau} < m/2$ \emph{and} $S_{< -\tau} < m/2$, then we see that $\med\Brace{X_1, \ldots, X_m} \in [-\tau, \tau]$.

For a random variable $X \sim \cau(0,1)$, we know that $\Pr(X \leq x) = 1/2 + \arctan(x) / \pi$.
For $0 < \tau < 1$, we see that $\Pr(X \geq \tau) = 1/2 - \arctan(\tau) / \pi \leq 1/2 - \tau / 4$.
By additive Chernoff bounds, we see that
\[
\Pr\Paren{S_{> \tau} \geq \frac{m}{2}}
\leq \exp\Paren{ -\frac{2 m^2 \tau^2}{16 m}}
= \exp\Paren{ -\frac{m \tau^2}{8} }
\]
Similarly, we have $\Pr\Paren{S_{< -\tau} \geq m/2} \leq \exp\Paren{ -m \tau^2/8 }$.
The claim follows from a union bound over the events $S_{> \tau} \geq m/2$ and $S_{< -\tau} \geq m/2$.
\end{proof}

\ignore{
\twophased*
\begin{proof}[Proof of \cref{thm:two-phased}]
The total sample complexity of \cref{alg:recovery-two-phased} is $m = m_1 + m_2$.

{\color{red}Update for batchleastsquares}

By \cref{thm:least-squares} and \cref{thm:cauchyest}, we know the following:
\yuhao{lots of question mark in this section}
Both \texttt{LeastSquares} and \texttt{CauchyEst} recover estimates $\wh{A}_1, \ldots, \wh{A}_n$ such that
\[
\Pr\Paren{ \forall i \in [n], \Abs{\Delta_i^\top M_i \Delta_i} \geq \sigma_i^2 \cdot \frac{\epsilon p_i}{n d_{avg}} } \leq \delta
\]
For \texttt{LeastSquares}, $m_1 = \frac{32n d_{avg}}{\epsilon} \cdot \ln\Paren{\frac{3n}{\delta}}$ samples suffices.
For \texttt{CauchyEst}, $m_1 = \frac{8 n d_{avg} d}{\epsilon} \log\Paren{\frac{2n}{\delta}}$ samples suffices.
Note that for polytrees, $n d_{avg} \leq n-1$.

Using the estimated coefficients from \texttt{LeastSquares} or \texttt{CauchyEst}, \cref{thm:recovery-variance} tells us that $m_2 = \frac{32n d_{avg}}{\epsilon} \log\Paren{\frac{2n}{\delta}}$ samples suffice to recover estimates $\wh{\sigma}_1, \ldots, \wh{\sigma}_n$ such that
\[
\Pr\Paren{\forall i \in [n], \Paren{1 - \sqrt{\frac{\epsilon p_i}{n d_{avg}}}} \cdot \sigma_i^2 \leq \wh{\sigma}_i^2 \leq \Paren{1 + \sqrt{\frac{\epsilon p_i}{n d_{avg}}}} \cdot \sigma_i^2} \geq 1 - \delta
\]

By \cref{thm:decomposition-implication}, we will recover an induced probability distribution $\mathcal{Q}$ such that $d_{KL}(\mathcal{P}, \mathcal{Q}) \leq 3\epsilon$.
\end{proof}}

\section{Median absolute deviation}
\label{sec:mad}

In this section, we give a pseudo-code of the well-known Median Absolute Deviation (MAD) estimator (see~\cite{huber2004robust} for example) which we use for the component-wise variance recovery in the contaminated setting.
The scale factor, $1/\Phi^{-1}(3/4) \approx 1.4826$ below, is needed to make the estimator unbiased.

\begin{algorithm}[htbp]
\caption{\texttt{MAD}: Variance recovery in the contaminated setting}
\label{algo:mad}
\begin{algorithmic}[1]
    \State \textbf{Input}: Contaminated samples $\{x_1,x_2,\dots,x_m\}$ from a univariate Gausssian
    \State $\wh{\mu}\gets \med\Brace{x_1,x_2,\dots,x_m}$.
    \State $\wh{\sigma}\gets 1.4826\cdot\med\Brace{|x_1-\wh{\mu}|,|x_2-\wh{\mu}|,\dots,|x_m-\wh{\mu}|}$.
    \State \Return $\wh{\sigma}$
\end{algorithmic}
\end{algorithm}

\end{document}